\documentclass[10pt]{article}
\usepackage{amssymb,amsmath,natbib,amsfonts,amsthm}
\usepackage[english]{babel}
\usepackage[ansinew]{inputenc}
\usepackage{graphics,epsfig}

\theoremstyle{plain}
\newtheorem{theorem}{Theorem}[section]
\newtheorem{lemma}[theorem]{Lemma}
\newtheorem{corollary}[theorem]{Corollary}

\newcommand{\argmax}{\operatornamewithlimits{argmax}}

\title{A Class of Solvable Optimal Stopping Problems of Spectrally Negative Jump Diffusions}

\author{Luis H. R. Alvarez E.\thanks{Department of Accounting and Finance, Turku School of Economics,
FIN-20014 University of Turku, Finland, E-mail: lhralv@utu.fi} \and Pekka Matomäki\thanks{Department of Accounting and Finance, Turku School of Economics,
FIN-20014 University of Turku, Finland, E-mail: pjsila@utu.fi}  \and
Teppo A. Rakkolainen\thanks{Veritas Pension Insurance, P.O. Box 133, FI-20700 Turku, Finland, E-mail: teppo.rakkolainen@gmail.com}}

\begin{document}

\maketitle

\begin{abstract}
We consider the optimal stopping of a class of spectrally
negative jump diffusions. We state a set of conditions under which
the value is shown to have a representation in terms of an ordinary
nonlinear programming problem. We establish a connection between the
considered problem and a stopping problem of an associated
continuous diffusion process and demonstrate how this connection may
be applied for characterizing the stopping policy and its value. We
also establish a set of typically satisfied conditions under which
increased volatility as well as higher jump-intensity decelerates
rational exercise by increasing the value and expanding the
continuation region.
\end{abstract}



\thispagestyle{empty} \clearpage \setcounter{page}{1}

\section{Introduction}
It is a well-known result from literature on mathematical finance
that the price of a perpetual American option on an underlying asset
whose value can be characterized as a stochastic process coincides
with the value of an optimal stopping problem for this process (see,
for example, \citet{KarandShr} pp. 54--87 and \citet{oksendal}, pp.
290--298). Such option prices, while naturally of interest in
themselves, can also be used as upper bounds for prices of American
options with finite expiration dates. Thus, their role is of
importance from a risk management point of view as well. Perpetual
optimal stopping problems arise quite naturally  also in the real
options literature on the valuation of irreversible investment
opportunities (see \citet{DixitandPindyck} for an extensive textbook
treatment of this theory, and \citet{boyarchenkolevendorskii2} for
some more recent developments in this field). In that modeling
framework the investment decision is usually interpreted as an
opportunity but not an obligation to obtain a stochastically
fluctuating return in exchange from a payment (sunk cost) which may
or may not be stochastic as well. Given the considerable planning
horizon of the valuation of real investment opportunities, the time
horizon is typically assumed to be infinite, i.e. the considered
optimal timing problem of the investment opportunity is assumed to
be perpetual.

When the dynamics of the underlying process are characterizable via
an Itô stochastic differential equation of form
\begin{eqnarray}
d X_{t}=\alpha(X_{t}) dt+\sigma(X_{t}) d W_{t}
\end{eqnarray}
with $W$ a standard Wiener process, the stopping problem has been
widely studied by relying on various techniques. Perhaps the most
common approach is to rely on variational inequalities or the
classical Hamilton-Jacobi-Bellman approach due to its applicability
in a multidimensional setting as well (cf. \citet{oksendal} and
\citet{oksendaletreikvam}). In the one-dimensional setting there
are, however, several different techniques for analyzing the
perpetual stopping problem. The most general approach is probably
provided by studies relying on the integral characterization of
excessive functions for diffusion processes and the Martin boundary
theory (cf. \citet{salminen} and \citet{BorandSalm}, pp. 32--35).
Alternatively, the considered problem can be analyzed by relying on techniques based on
martingales like the {\em Snell envelope} (see \citet{peskiretshiryaev} for a thorough characterization and comprehensive list of references), or the {\em Beibel-Lerche}
approach (see, for example, \citet{BeLe1} and \citet{LeUr}). One dimensional stopping problems have been also analyzed by exploiting
the relationship between functional concavity and $r$-excessivity
along the lines of the pioneering work by \citet{Dynkin} (Chapters
XV and XVI) and by \citet{DynkinandYus} which has been subsequently
applied within a general optimal stopping framework in
\citet{dayanikkaratzas}. An approach based on occupation measures
and infinite dimensional linear programming is, in turn, developed in
\citet{ChSto} and \citet{HeSto}. Another technique for studying the
perpetual optimal stopping problem in the linear diffusion setting
is provided by the approaches relying on the well-known relationship
between excessivity and superharmonicity with respect to first exit
times from open sets with compact closure in the state space of the
considered diffusion (cf. \citet{Dynkin}, Theorem 12.4). In this
case, the optimal stopping problem is reduced to the optimization of
arbitrary boundaries and can be analyzed by relying on
ordinary nonlinear programming techniques (cf. \citet{alvarez1} and
\citet{alvarez2}).

More recently, the shortcomings of continuous path models driven by
a Brownian motion have become apparent and, consequently, more
general models allowing path discontinuities have been studied. In
many ways the most simple generalizations of the traditional
continuous path models are jump diffusion models, in which the
driving noise is a L\'{e}vy process. L\'{e}vy processes can be used
to construct more realistic models of financial quantities, as they
are able to generate jump discontinuities and the leptokurtic
feature of return distributions, unlike the Gaussian models based on
a Brownian motion and the normal distribution. For a taste of the
plentiful research done on pricing American options and optimal
stopping in L\'{e}vy models, see (for example) \citet{alili}, \citet{boyarchenkolevendorskii},
\citet{boyarchenkolevendorskii3}, \citet{duffie}, \citet{gerber1},
\citet{gerber2},  \citet{Kou}, \citet{mordecki1},
\citet{mordecki2}
\citet{mordeckietsalminen}, and  \citet{pham}.

In risk management a criticism often leveled against the continuous
models is their inability to model downside risk: the possibility of
an instantaneous drop in the value of an asset. In real life markets
phenomena closely resembling such instantaneous drops are often
observed (for example, sudden unanticipated deterioration of stock
market values, credit defaults, etc.). An empirically observed fact
is that in the stock market reactions to negative shocks are usually
significantly stronger than the reactions to positive ones (this is
the celebrated "bad news" principle originally introduced in the
seminal study by \citet{bernanke}). In light of this asymmetric
nature of the reaction to unanticipated shocks, a prudent approach
is to disregard possibilities for positive surprises and to take
fully into account the possibilities for disadvantageous future
occurrences. Consequently, a one-sided model that allows
instantaneous downward jumps can be seen as an acceptable model from
a prudent risk management point of view.

Motivated by our previous arguments, it is our objective in this
study to consider a spectrally negative one-dimensional jump
diffusion, say $X$, with a state space $\mathcal{I}=(a,b)\subseteq \mathbb{R}$
and unattainable boundaries $a$ and $b$. Interestingly, we establish that
given some extra conditions on $X$, the value of the optimal
stopping problem has a representation in terms of an ordinary
nonlinear programming problem (cf. \citet{alvarez1} and
\citet{alvarez4} for an associated result in the continuous
diffusion case).

In order to develop relatively easily verifiable sufficient conditions we consider an optimal stopping problem
of an associated continuous diffusion process which can be obtained
by removing the pure jump part of the considered L\'{e}vy diffusion.
We demonstrate that the value of the considered jump-diffusion
stopping problem can be "sandwiched" between the values of two
stopping problems which are defined with respect to the associated
continuous diffusion. This finding is of interest since it can be
applied for deriving bounds for the exercise threshold of the
considered optimal stopping problem for the underlying
jump-diffusion. Moreover, since the restricting values defined with
respect to the continuous diffusion differ only by the rate at which
they are discounted, our findings indicate that under some
circumstances the downside jump-risk can be directly incorporated
into the continuous diffusion case by adjusting the discount rate
appropriately (for some results in this direction, see
\citet{AlvarezRakkolainen2010}). This characterization is also
important in the analysis of the impact of downside risk on the
optimal stopping policy since according to this representation the
optimal exercise boundary is lower for the underlying jump-diffusion
than for the associated dominating continuous diffusion process
provided that both valuations are discounted at the same rate.

We also investigate the comparative static properties of the optimal
stopping policy and its value and present a set of relatively
general conditions under which the value of the considered problem
is convex. Along the lines of previous studies
considering the optimal stopping of linear diffusions, we find that
in such a case higher volatility increases the value of the optimal
strategy and expands the continuation region where stopping is
suboptimal by increasing the optimal exercise threshold. These
observations are of interest since they indicate that higher
volatility decelerates the rational exercise of investment
opportunities by increasing the option value of waiting in the
presence of jumps as well. Interestingly, our results also indicate that
higher jump-intensity increases the value
of waiting and decelerates rational exercise by expanding the
continuation region. These observations emphasize the potentially
significant combined negative effect of jump-risk and continuous
systematic risk on the timing of irreversible investment policies.

We also analyze if the considered class of stopping problems can be
expressed in terms of an associated deterministic timing problem. Somewhat surprisingly, we
find that the value of the considered optimal stopping problem coincides with the value of an
associated optimal timing problem of a deterministic process evolving according to the dynamics characterized by an
ordinary first order differential equation adjusted
to the risk generated by the driving Lévy process. We show that this representation is valid for single threshold policies
within the considered class of stopping problems. Since an analogous certainty equivalent formulation is
valid in the continuous diffusion setting as well, our findings can be utilized in decomposing the required exercise premium
into two parts: a part based on the continuous Brownian dynamics and a part based on the discontinuous compensated compound Poisson process.

The contents of this study are as follows. In section 2, we present
the model and the basic assumptions used throughout the study. The auxiliary results based
on the associated continuous diffusion are then developed in section 3. The
representation of the stopping problem in terms of an ordinary
optimization problem is then stated and proved in section 4. The sensitivity of the optimal exercise policy and its value with respect to changes in volatility and the jump intensity of the driving dynamics are then investigated in section 5. Section 6 in turn summarizes our results on the certainty equivalent valuation principles. Explicit illustrations
are given in section 7, and section 8 finally concludes our study.

\section{The Setup and Basic Assumptions}

Let $(\Omega,\mathcal{F},\mathbb{P})$ be a probability space
carrying a standard Wiener process $W=\{W_{t}\}$ and a compound
Poisson process $J=\{J_{t}\}$ with intensity $\lambda$ and a jump
size distribution on $\mathcal{S}\subseteq \mathbb{R}$ characterized
by a probability distribution $\mathfrak{m}$. We define a L\'{e}vy
process $L=\{L_{t}\}$ by
\begin{eqnarray} \label{levyprocess}
L_{t}=  t + W_{t} + J_{t}.
\end{eqnarray}
We equip $(\Omega,\mathcal{F},\mathbb{P})$ with the completed
natural filtration $\mathbb{F}$ generated by this process. The
natural filtration of a L\'{e}vy process is right-continuous, and
thus the completed filtration satisfies the usual hypotheses (see
\citet{protter} Theorem I.31).

Given the driving L\'{e}vy
process, let $X=\{X_{t}\}$ be a jump diffusion evolving according to the dynamics characterized by the stochastic
differential equation
\begin{eqnarray} \label{jumpdiffusion}
d X_{t}= \alpha(X_{t-}) dt + \sigma(X_{t-}) d W_{t} +
\int_{\mathcal{S}} \gamma(X_{t-},z) \tilde{N}(dz,dt),\quad X_{0}=x\in \mathcal{I},
\end{eqnarray}
where $\tilde{N}(U,t)$ is a compensated Poisson
random measure with characteristic (L\'{e}vy) measure
$\nu(dz)=\lambda  \mathfrak{m}(dz)$ and $\mathfrak{m}$ is a
probability measure on $\mathcal{S}$. Note that the driving jump
process is, as a compensated process, a martingale -- this is no
restriction, as non-martingale jump dynamics can be reduced to the
form \eqref{jumpdiffusion} by adding and subtracting a correction term
on the left side of the stochastic differential equation. We denote
the expectation of the jump size by $\overline{m}$. The state space
of the L\'{e}vy diffusion is an open interval $\mathcal{I}:=(a,b)\subseteq
\mathbb{R}$ where $a$ and $b$ are unattainable boundaries (not attainable
in finite time). We assume that the coefficient functions in
\eqref{jumpdiffusion} satisfy some sufficient conditions for the
existence of a unique adapted c\`{a}dl\`{a}g solution $X \in
L^{2}(\mathbb{P})$ without explosions; in the case of an infinite
interval $\mathcal{I}$, the usual sufficient conditions are at most linear
growth and Lipschitz continuity, see \citet{oksendaletsulem} Theorem
1.19. The global Lipschitz condition guarantees that the explosion
time of the process is a.s. infinite (see \citet{protter} Theorem
V.40) and that $X$ is strong Markov (cf. \citet{protter} Theorem
V.32). Observe also that the jump times of $X$ coincide with the
jump times of the driving L\'{e}vy process, which are totally
inaccessible stopping times (cf. \citet{protter} Theorem III.4).
This implies that $X$ is quasi-left continuous (or left-continuous
over stopping times) and hence a Hunt process. To avoid
degeneracies, we also assume that the volatility coefficient satisfies the inequality
$\sigma(x) >0$ on $\mathcal{I}$. Finally, we assume that
$a-x < \gamma(x,z)\leq 0$ for all $(x,z) \in \mathcal{I} \times
\mathcal{S}$. This assumption guarantees that $X$ has only negative jumps and
cannot reach the lower boundary $a$ by jumping.

Our objective is to consider the optimal stopping
problem
\begin{eqnarray} \label{stoppingproblem}
V_\lambda(x)=\sup_{\tau \in \mathcal{T}} \mathbb{E}_{x}\left[ e^{-r \tau}g(X_{\tau})\right],
\end{eqnarray}
where  $\mathcal{T}$ denotes the set
of all $\mathbb{F}$ -stopping times and $g:\mathcal{I}\mapsto \mathbb{R}$ denotes the exercise payoff.
We assume that the exercise payoff is continuous and nondecreasing on $\mathcal{I}$ and that there is a unique break even state $g^{-1}(0)=x_0\in \mathcal{I}$ such that $g(x)>0$ for all $x\in (x_0,b)$. We also assume that $g\in C^1(\mathcal{I}\backslash \mathcal{N})\cap C^2(\mathcal{I}\backslash \mathcal{N})$, $g'(x\pm) < \infty$, and $|g''(x\pm)|<\infty$ for all $x\in \mathcal{N}$, where $\mathcal{N}$ is a finite set of points in $\mathcal{I}$.

The solution of the optimal stopping problem is known to be closely
related to the integro-differential equation (cf. \citet{oksendaletsulem}) defined for $f \in
C_{0}^{2}(\mathcal{I})$ by
\begin{eqnarray} \label{integrode1}
\mathcal{G}f = r f,
\end{eqnarray}
where $(\mathcal{G}f)(x)$ is the generator of $X$ given by
\begin{eqnarray} \label{generator}
\frac{1}{2} \sigma^{2}(x) f^{\prime
\prime}(x)+\alpha(x)f^{\prime}(x)+\lambda
\int_{\mathcal{S}}\left\{f(x+\gamma(x,z))-f(x)-f^{\prime}(x)\gamma(x,z)\right\}\mathfrak{m}(dz).
\end{eqnarray}
Integrating the last two terms of the integrand in \eqref{generator}
and using the notation $(\mathcal{G}_{r}):=(\mathcal{G}-r)$ we can
write \eqref{integrode1} equivalently as
\begin{eqnarray*}
(\mathcal{G}_{r} f)(x)=\frac{1}{2} \sigma^{2}(x) f''(x) + \tilde{\alpha}(x)f'(x)-(r+\lambda)f(x)+\lambda
\int_{\mathcal{S}}f(x+\gamma(x,z))\mathfrak{m}(dz)=0,
\end{eqnarray*}
where $\tilde{\alpha}(x)=\alpha(x)-\lambda \int_{\mathcal{S}}
\gamma(x,z)\mathfrak{m}(dz)$. We need to make the following assumption
on the integro-differential operator $\mathcal{G}_{r}$:
\begin{description}
\item[(A1)] The integro-differential equation $\mathcal{G}_{r} \psi_\lambda = 0$ has an increasing solution $\psi_\lambda \in \mathcal{C}^{2}(\mathcal{I})$.
\end{description}
\vskip .1in
It should be mentioned here that it is not at all clear whether a
given integro-differential equation has such a smooth solution --
the validity of this assumption needs to be checked in each
case.

Having stated our assumption on the existence of a positive solution of the integro-differential equation
$\mathcal{G}_{r} \psi_\lambda = 0$ we now can state the following auxiliary lemma.
\begin{lemma} \label{exittime}
Assume that condition {\rm (}{\bf A1}{\rm )} is met and denote as $\tau_{(a,y)}=\inf\{t\geq 0: X_t\not \in (a,y)\}$ the first exit time of the underlying jump diffusion from the set $(a,y)$, where $y\in \mathcal{I}$. Then for all $x < y$ we have
\begin{eqnarray} \label{Laplace}
\mathbb{E}_{x}[e^{-r \tau_{(a,y)}}]=\frac{\psi_\lambda(x)}{\psi_\lambda(y)}
\end{eqnarray}
and $\psi$ is increasing. Moreover, in case $\psi_\lambda(x)$ exists any
other nonnegative and increasing solution of $\mathcal{G}_{r}u=0$ is
a constant multiple of $\psi_\lambda(x)$ (i.e. $\psi_\lambda(x)$ is unique up to a
multiplicative constant).
\end{lemma}
\begin{proof}
Applying Dynkin's formula to the mapping
$(t,x)\mapsto e^{-rt}\psi_\lambda(x)$ yields
\begin{eqnarray*}
\mathbb{E}_{x}[e^{-r \tau_{(a,y)}}
\psi_\lambda(X_{\tau_{(a,y)}})]=\psi_\lambda(x)+\mathbb{E}_{x}\int_{0}^{\tau_{(a,y)}}e^{-rt}(\mathcal{G}_{r}\psi_\lambda)(X_{t})dt.
\end{eqnarray*}
Since $\psi_\lambda$ solves $\mathcal{G}_{r}\psi_\lambda=0$ and $X_{\tau_{(a,y)}}=y$
a.s. (because $X$ has no positive jumps and it never attains $a$),
this implies that
\begin{eqnarray*}
\psi_\lambda(y) \mathbb{E}_{x}[e^{-r \tau_{(a,y)}}]=\psi_\lambda(x),
\end{eqnarray*}
from which the first two of the claimed results follow (for the
latter one, note that $\mathbb{E}_{x}[e^{-r \tau_{(a,y)}}] \in
(0,1]$). To establish uniqueness, assume that $\varsigma:\mathcal{I}\mapsto
\mathbb{R}_+$ is another increasing and nonnegative solution of
equation $\mathcal{G}_{r}u=0$. By applying a similar argument as
above, we find that
$$
\varsigma(x) = \frac{\varsigma(y)}{\psi_\lambda(y)}\psi_\lambda(x)
$$
which completes the proof of our lemma.
\end{proof}

It is worth emphasizing that the strong Markov property of the jump
diffusion and the fact that it can increase only continuously imply that the function
$\mathbb{E}_{x}[e^{-r \tau_{(a,y)}}]$ can always be expressed as a
ratio of the form \eqref{Laplace}. However, it is not beforehand
clear whether this ratio is always (i.e. for any jump diffusion
model) twice continuously differentiable with respect to the current
state or not. Hence, lemma \ref{exittime} essentially demonstrates
that in those cases where the integro-differential equation
$\mathcal{G}_{r}u=0$ has an increasing solution, the expected value
$\mathbb{E}_{x}[e^{-r \tau_{(a,y)}}]$ can be expressed in terms of
this solution and identity \eqref{Laplace} holds.

It is also worth noticing that the function $\psi_\lambda(x)$ is related to the more general class of functions known as the $r$-scale
functions familiar from the literature on Lévy processes (for an excellent treatment, see Chapter 8 in \citet{kyprianou}).
As known from that literature, the Laplace transform of the first exit time from the set $(a,y)\subset \mathcal{I}$ can always be expressed as
\begin{eqnarray}
\mathbb{E}_{x}[e^{-r \tau_{(a,y)}}] = \frac{F^{(r)}(x)}{F^{(r)}(y)},\label{rscale}
\end{eqnarray}
where $F^{(r)}(x)$ denotes the continuously differentiable  $r$-scale function associated to the underlying spectrally negative jump diffusion (cf. Theorem 8.1 in \citet{kyprianou}). Hence, Lemma \ref{exittime} essentially states that if a function $\psi_\lambda(x)$ satisfying assumption ({\bf A1}) exists, it has to coincide with the $r$-scale function $F^{(r)}(x)$.

Finally, we wish to point out here that in \citet{KouWang2003},
representation results similar to our lemma \ref{exittime} are
obtained for a Brownian motion augmented with a compound Poisson
process with double exponentially distributed jumps.

\section{Sandwiching the Solution}

In this section we plan to develop auxiliary inequalities based on
two optimal stopping problems of an associated
continuous diffusion model. To accomplish this task, consider now
the associated diffusion
\begin{eqnarray} \label{associateddiffusion}
d\tilde{X}_{t}:=
\tilde{\alpha}(\tilde{X}_{t})dt+\sigma(\tilde{X}_{t})dW_{t},\quad \tilde{X}_0=x.
\end{eqnarray}
Given the associated diffusion process $\tilde{X}$ we now introduce the associated stopping problem
\begin{eqnarray} \label{associatedprob}
\tilde{V}_\theta(x)=\sup_{\tilde{\tau}} \mathbb{E}_{x}\left[ e^{-\theta \tilde{\tau}}
g(\tilde{X}_{\tilde{\tau}})\right]
\end{eqnarray}
and denote as $\tilde{C}_\theta=\{x\in \mathcal{I}:\tilde{V}_\theta(x)>g(x)\}$ the continuation region and as
$\tilde{\Gamma}_\theta=\{x\in \mathcal{I}:\tilde{V}_\theta(x)=g(x)\}$ the stopping region associated to \eqref{associatedprob}.
It is worth mentioning that the associated diffusion is very useful
in assessing the impact of downside risk on the optimal policy, as
the Lévy diffusion $X$ is, in fact, a superposition of $\tilde{X}$
and a spectrally negative, nondecreasing, non-martingale jump process. In accordance with our assumptions concerning the boundary behavior of the
jump diffusion $X$, we now assume that the boundaries of the state space of $\tilde{X}$ are natural.

As usually,
we denote as $\tilde{\mathcal{A}}_\theta$ the differential operator
\begin{eqnarray*}
\tilde{\mathcal{A}}_\theta =
\frac{1}{2}\sigma^{2}(x)\frac{d^2}{dx^2} +
\tilde{\alpha}(x)\frac{d}{dx}
-\theta
\end{eqnarray*}
associated with the continuous diffusion $\tilde{X}_t$ killed at the
constant rate $\theta>0$ and as
$$
S'(x) = \exp\left(-\int\frac{2\tilde{\alpha}(x)dx}{\sigma^2(x)}\right)
$$
the density of the scale function of the diffusion
$\tilde{X}_t$.

Along the lines of the notation in our
previous analysis, we denote as $\tilde{\psi}_\theta(x)$ the
increasing fundamental solution of the ordinary linear second order
differential equation $(\tilde{\mathcal{A}}_\theta u)(x)=0$ (for a
thorough characterization of these mappings and their boundary behavior, see
\citet{BorandSalm}, pp. 18--19). As is
well-known from the classical theory of diffusions, given this
increasing fundamental solution we have for all $x\leq y$ (cf.
\citet{BorandSalm}, p. 18)
$$
\mathbb{E}_{x}\left[e^{-\theta \tilde{\tau}_{(a,y)}}\right] =
\frac{\tilde{\psi}_{\theta}(x)}{\tilde{\psi}_{\theta}(y)},
$$
where  $\tilde{\tau}_{(a,y)}=\inf\{t\geq 0: \tilde{X}_t \not\in (a,y)\}$
denotes the first exit time of the diffusion $\tilde{X}_t$ from the
set $(a,y)$. Therefore, the continuity of the exercise payoff yields
that for all $x\leq y$ we have
$$
\mathbb{E}_{x}\left[e^{-\theta
\tilde{\tau}_{(a,y)}}g(\tilde{X}_{\tilde{\tau}_{(a,y)}})\right] =
g(y)\frac{\tilde{\psi}_{\theta}(x)}{\tilde{\psi}_{\theta}(y)}
$$
implying that
$$
\sup_{y\geq x}\mathbb{E}_{x}\left[e^{-\theta
\tilde{\tau}_{(a,y)}}g(\tilde{X}_{\tilde{\tau}_{(a,y)}})\right] =
\tilde{\psi}_{\theta}(x)\sup_{y\geq
x}\left[\frac{g(y)}{\tilde{\psi}_{\theta}(y)}\right]
$$
provided that the supremum exists.

Before proceeding in our analysis, we first establish the following result
establishing that in the present setting the value $V_\lambda(x)$ can be sandwiched between the values
of two
associated stopping problems of the associated diffusion $\tilde{X}$.
\begin{lemma}\label{domrexcessive}
Assume that the function $f:\mathcal{I}\mapsto \mathbb{R}_+$ is non-decreasing.
\begin{enumerate}
  \item[(A)] if $f$ is $r$-excessive for the the diffusion $\tilde{X}_t$ then
it is $r$-excessive for the jump diffusion $X_t$ as well.
  \item[(B)] if $f$ is $r$-excessive for the the jump diffusion $X_t$ then
it is $r+\lambda$-excessive for the diffusion $\tilde{X}_t$ as well.
\end{enumerate}
Therefore,
$\tilde{V}_{r+\lambda}(x)\leq V_\lambda(x) \leq \tilde{V}_{r}(x)$
for all $x\in \mathcal{I}$, $\tilde{C}_{r+\lambda}\subseteq \{x\in \mathcal{I}:V_\lambda(x)>g(x)\}\subseteq \tilde{C}_{r}$, and
$\tilde{\Gamma}_{r}\subseteq  \{x\in \mathcal{I}:V_\lambda(x)=g(x)\} \subseteq \tilde{\Gamma}_{r+\lambda}$
\end{lemma}
\begin{proof}
(A) It is clear by the definition of the jump diffusion $X_t$ and the associated diffusion $\tilde{X}_t$ that $X_t\leq \tilde{X}_t$ a.s.
Assume now that the function $f:\mathcal{I}\mapsto \mathbb{R}_+$ is non-decreasing and $r$-excessive for the the diffusion $\tilde{X}_t$.
We then have that $f(x)\geq \mathbb{E}_x[e^{-r t}f(\tilde{X}_t)] \geq \mathbb{E}_x[e^{-r t}f(X_t)]$ for all $(t,x)\in \mathbb{R}_+\times \mathcal{I}$ demonstrating that $f$ is $r$-excessive for $X$ as well.

(B) Consider now the second inequality and denote as $T_\lambda$ the first exponentially distributed date at which the driving compound Poisson process experiences a jump. If $f$ is $r$-excessive for $X_t$ then we have for all $(t,x)\in \mathbb{R}_+\times \mathcal{I}$
\begin{eqnarray*}
f(x)\geq \mathbb{E}_x[e^{-rt}f(X_t)] &=& \mathbb{E}_x[e^{-rt}f(\tilde{X}_t)\mathbf{1}_{\{t<T_\lambda\}} + e^{-rt}f(X_t)\mathbf{1}_{\{t\geq T_\lambda\}}]\\
&=& \mathbb{E}_x[e^{-(r+\lambda)t}f(\tilde{X}_t)] + \mathbb{E}_x[e^{-rt}f(X_t)\mathbf{1}_{\{t\geq T_\lambda\}}]\\
&\geq& \mathbb{E}_x[e^{-(r+\lambda)t}f(\tilde{X}_t)]
\end{eqnarray*}
since $T_\lambda$ is independent of the driving Brownian motion, $X_t=\tilde{X}_t$ for $t<T_\lambda$, and $f(x)$ is nonnegative. Hence, we find that $f(x)$ is $r+\lambda$-excessive for
$\tilde{X}_t$ as well.

It remains to consider the ordering of the values of the considered stopping problems. We first observe that the assumed monotonicity of the exercise payoff implies that the value of the optimal stopping strategy is monotonic as well. However, since the value $V_\lambda(x)$ constitutes the least $r$-excessive majorant of the payoff $g(x)$ for the jump diffusion $X$ and $\tilde{V}_\theta(x)$ constitutes the least $\theta$-excessive majorant of the payoff $g(x)$ for the diffusion $\tilde{X}$ the alleged ordering follows from part (A) and (B). Our results on the continuation regions and stopping regions of the considered stopping problems are now straightforward implications of the proven ordering.
\end{proof}

Lemma \ref{domrexcessive} demonstrates that the value of the optimal stopping problem of the
jump diffusion can be sandwiched between the values of two associated stopping problems of
the associated continuous diffusion process $\tilde{X}$.  More precisely, Lemma \ref{domrexcessive} proves that
$\tilde{V}_{r+\lambda}(x)\leq V_\lambda(x) \leq \tilde{V}_r(x)$ for all $x\in \mathcal{I}$.
This observation is
interesting since it directly generates a natural ordering for the
monotone and smooth solutions of the variational inequalities
$\max\{(\mathcal{G}_ru)(x), g(x)-u(x)\}=0$ and
$\max\{(\mathcal{\tilde{A}}_{\theta}u)(x),g(x)-u(x)\}=0$ with
$\theta=r, r+\lambda$.

In light of the observation of Lemma \ref{domrexcessive} it
is naturally of interest to ask whether the discount rate $\theta$
can be chosen so as to extend the findings of  Lemma \ref{domrexcessive} to the expected present values of a unit of account at exercise.
 A set of results indicating that such ordering holds for a class of nondecreasing
cases considered in this study are summarized in our next lemma.
\begin{lemma}\label{domination}
Denote as $\tilde{M}_t=\sup\{\tilde{X}_s; s\leq t\}$ and as $M_t=\sup\{X_s; s\leq t\}$ the running maximum processes of $\tilde{X}$ and $X$, respectively, and let $T_\theta\sim\exp(\theta)$ be an exponentially distributed random date independent of $X$ and $\tilde{X}$. We then have
\begin{eqnarray}
\tilde{M}_{T_{r}\wedge T_{\lambda}}\leq M_{T_r}\leq \tilde{M}_{T_r},\quad \textrm{a.s.}\label{supineq}
\end{eqnarray}
Moreover,
\begin{eqnarray}
\frac{\tilde{\psi}_{r+\lambda}(x)}{\tilde{\psi}_{r+\lambda}(y)}\leq
\mathbb{E}_{x}\left[e^{-r \tau_{(a,y)}}\right] \leq
\frac{\tilde{\psi}_r(x)}{\tilde{\psi}_r(y)}\label{Moment}
\end{eqnarray}
for all $x\leq y$,
$$\lim_{x\downarrow a}\mathbb{E}_{x}\left[e^{-r \tau_{(a,y)}}\right] = 0,$$ and
$$
\tilde{\psi}_{r+\lambda}(x)\sup_{y\geq
x}\left[\frac{g(y)}{\tilde{\psi}_{r+\lambda}(y)}\right]\leq
\sup_{y\geq x}\mathbb{E}_{x}\left[e^{-r
\tau_{(a,y)}}g(X_{\tau_{(a,y)}})\right] \leq
\tilde{\psi}_{r}(x)\sup_{y\geq
x}\left[\frac{g(y)}{\tilde{\psi}_{r}(y)}\right]
$$
provided that the suprema exist.
\end{lemma}
\begin{proof}
We first observe that since $X_t = \tilde{X}_t-\eta_t$, where $\eta_t$ is a spectrally negative, nondecreasing, and non-martingale jump process, we naturally have
that $X_t\leq \tilde{X}_t$ and, therefore, that $M_t\leq \tilde{M}_t$ a.s.. Consequently, we observe that $M_{T_r}\leq \tilde{M}_{T_r}$ a.s. as well proving \eqref{supineq}. On the other hand, since
$T_r\wedge T_\lambda \leq T_r$ and the running supremum process is nondecreasing, we also observe that $M_{T_r}\geq M_{T_{r}\wedge T_{\lambda}}=\tilde{M}_{T_{r}\wedge T_{\lambda}}$ a.s. Noticing now that $T_\lambda\wedge T_r\sim T_{r+\lambda}$ and applying the inequality \eqref{supineq} and the identities
$\mathbb{P}_x[\tilde{M}_{T_\theta}\geq y]=\mathbb{P}_x[\tilde{\tau}_y < T_\theta]=\mathbb{E}_x[e^{-\theta\tilde{\tau}_y}]$
and
$\mathbb{P}_x[M_{T_\theta}\geq y]=\mathbb{P}_x[\tau_y < T_\theta]=\mathbb{E}_x[e^{-\theta\tau_y}]$
then proves \eqref{Moment}. The rest of the alleged results then follow from the
nonnegativity of $g(x)$ and the fact $\lim_{x\downarrow a}\tilde{\psi}_{\theta}(x)=0$ for a natural boundary.
\end{proof}

Lemma \ref{domination} states a simple sufficient condition under which the
expected payoff accrued from following a standard one-sided threshold policy
can be sandwiched between two values defined with respect to the associated
continuous diffusion process. Since both bounding values can be under certain
circumstances identified as the values of an optimal stopping problem of the
associated continuous diffusion, Lemma \ref{domination} essentially states a sufficient
condition under which the maximal value which can be attained by following a single threshold
strategy is confined between the above mentioned two values.

A set of important implications of Lemma \ref{domination} applicable in the analysis of the first order condition characterizing the optimal exercise threshold
is summarized in our next corollary.
\begin{corollary}\label{dominationcorollary}
Assume that condition {\rm (}{\bf A1}{\rm )} is satisfied. Then,
\begin{eqnarray}
\frac{\tilde{\psi}_{r+\lambda}'(x)}{\tilde{\psi}_{r+\lambda}(x)}\geq
\frac{\psi_\lambda'(x)}{\psi_\lambda(x)} \geq
\frac{\tilde{\psi}_r'(x)}{\tilde{\psi}_r(x)}\label{logder}
\end{eqnarray}
for all $x\in \mathcal{I}$ and
\begin{eqnarray}
g'(x)\frac{\tilde{\psi}_{r+\lambda}(x)}{\tilde{\psi}_{r+\lambda}'(x)}-g(x)\leq
g'(x)\frac{\psi_\lambda(x)}{\psi_\lambda'(x)} - g(x) \leq
g'(x)\frac{\tilde{\psi}_r(x)}{\tilde{\psi}_r'(x)}-g(x)\label{nc1}
\end{eqnarray}
for all $x\in \mathcal{I}\backslash \mathcal{N}$.
Moreover, for all $x\in \mathcal{I}\backslash \mathcal{N}$ it holds that
\begin{eqnarray}
\frac{(\mathcal{L}_{\tilde{\psi}_{r+\lambda}} g)(x)}{\tilde{\psi}_{r+\lambda}'(x)}\leq \frac{(\mathcal{L}_{\psi_\lambda} g)(x)}{\psi_\lambda'(x)}  \leq \frac{(\mathcal{L}_{\tilde{\psi}_r} g)(x)}{\tilde{\psi}_{r}'(x)},\label{Martin}
\end{eqnarray}
where
\begin{eqnarray}
(\mathcal{L}_{u} g)(x) = \frac{g'(x)}{S'(x)}u(x)-\frac{u'(x)}{S'(x)}g(x).\label{Martin1}
\end{eqnarray}
\end{corollary}
\begin{proof}
Lemma \ref{domination} implies that for all $x\leq y$
$$
1-\frac{\tilde{\psi}_{r+\lambda}(x)}{\tilde{\psi}_{r+\lambda}(y)}\geq
1- \frac{\psi_\lambda(x)}{\psi_\lambda(y)} \geq
1- \frac{\tilde{\psi}_r(x)}{\tilde{\psi}_r(y)}
$$
which, in turn, implies that
$$
\int_{x}^{y}\frac{\tilde{\psi}_{r+\lambda}'(t)}{\tilde{\psi}_{r+\lambda}(y)}dt\geq
\int_{x}^{y}\frac{\psi_\lambda'(t)}{\psi_\lambda(y)}dt \geq
\int_{x}^{y}\frac{\tilde{\psi}_r'(t)}{\tilde{\psi}_r(y)}dt.
$$
Applying the mean value theorem for integrals and letting $x\uparrow y$ then proves inequality \eqref{logder}.
Inequality \eqref{nc1} then follows from the monotonicity of the reward payoff and inequality \eqref{logder}.
Finally, inequality \eqref{Martin} follows from inequality \eqref{nc1} after noticing that
$$
(\mathcal{L}_{\psi_\lambda} g)(x) = \frac{\psi_\lambda'(x)}{S'(x)}\left[\frac{g'(x)}{\psi_\lambda'(x)}\psi_\lambda(x)-g(x)\right].
$$
\end{proof}

\noindent Corollary \ref{dominationcorollary} essentially shows that if condition ({\bf A1}) is satisfied, then the logarithmic
growth rates of the fundamental solutions are ordered. This result is important, since it implies that the ratio
$g(x)/\psi_\lambda(x)$ is decreasing on the set where the ratio $g(x)/\tilde{\psi}_r(x)$ is decreasing and increasing on the set where
the ratio $g(x)/\tilde{\psi}_{r+\lambda}(x)$ is increasing. Given the representation \eqref{rscale}, we notice that the results of Corollary \ref{dominationcorollary} are satisfied by the $r$-scale function $F^{(r)}(x)$ as well since the underlying process is of unbounded variation
(by Lemma 8.2. in \citet{kyprianou}).
Moreover, since
\begin{eqnarray}
(\mathcal{L}_{\tilde{\psi}_\theta} g)'(x) = (\tilde{\mathcal{A}}_\theta g)(x)\tilde{\psi}_{\theta}(x)m'(x),\quad x\in \mathcal{I}\backslash \mathcal{N},\label{operator}
\end{eqnarray}
we observe the monotonicity of the ratio $g(x)/\tilde{\psi}_\theta(x)$ is
essentially dictated by the properties of the mapping $(\tilde{\mathcal{A}}_\theta g)(x)$.

We now present the following auxiliary result needed later in the proof of the existence of a unique optimal stopping boundary in the jump diffusion setting.
\begin{lemma}\label{l1}
Assume that there is a unique $\hat{x}_\theta > x_0$ so that
$$
(\tilde{\mathcal{A}}_{\theta}g)(x) \gtrless 0,\quad x \lessgtr \hat{x}_\theta, x\not\in \mathcal{N}.
$$
Assume also that $g'(x+)\geq g'(x-)$ for all $x\in (x_0,\hat{x}_\theta)\cap\mathcal{N}$ and $g'(x+)\leq g'(x-)$ for all $x\in (\hat{x}_\theta,b)\cap\mathcal{N}$. Then, there is a unique maximizing threshold
\begin{eqnarray}
x_\theta^\ast = \argmax\left\{\frac{g(x)}{\tilde{\psi}_\theta(x)}\right\} \in [\hat{x}_\theta,b)\label{threshold}
\end{eqnarray}
such that $\tilde{\tau}_{x_\theta^\ast}=\inf\{t\geq 0: \tilde{X}_t\geq x_\theta^\ast\}$ constitutes the optimal stopping strategy of \eqref{associatedprob} and
\begin{eqnarray}
\tilde{V}_\theta(x) = \tilde{\psi}_{\theta}(x)\sup_{y\geq x}\left\{\frac{g(y)}{\tilde{\psi}_{\theta}(y)}\right\} = \begin{cases}
g(x) &x\geq x_\theta^\ast\\
\tilde{\psi}_{\theta}(x)\frac{g(x_\theta^\ast)}{\tilde{\psi}_{\theta}(x_\theta^\ast)} &x < x_\theta^\ast.
\end{cases}\label{value}
\end{eqnarray}
\end{lemma}
\begin{proof}
Consider the behavior of the mapping $(\mathcal{L}_{\tilde{\psi}_\theta} g)(x)$ on $\mathcal{I}$. We first observe that the monotonicity and non-positivity of the exercise payoff $g(x)$ on $(a,x_0)$ guarantee that $(\mathcal{L}_{\tilde{\psi}_\theta} g)(x) \geq 0$ on $(a,x_0)$ and $(\mathcal{L}_{\tilde{\psi}_\theta} g)(x_0+)=\frac{g'(x_0+)}{S'(x_0)}\psi(x_0) \geq0$. Second, applying \eqref{operator}, the definition of $(\mathcal{L}_{\tilde{\psi}_\theta} g)(x)$, and invoking our assumption on the local behavior of $g'(x)$ on $\mathcal{N}$ shows that $(\mathcal{L}_{\tilde{\psi}_\theta} g)'(x)>0$ for all $x\in (x_0,\hat{x}_\theta)\backslash \mathcal{N}$ and
$$
(\mathcal{L}_{\tilde{\psi}_\theta} g)(x+)=(\mathcal{L}_{\tilde{\psi}_\theta} g)(x-)+\frac{\tilde{\psi}_{\theta}(x)}{S'(x)}\left[g'(x+)-g'(x-)\right]\geq 0
$$
for all $x\in (x_0,\hat{x}_\theta)\cap \mathcal{N}$. In a completely analogous fashion, we find that $(\mathcal{L}_{\tilde{\psi}_\theta} g)'(x)<0$ for all $x\in (\hat{x}_\theta,b)\backslash \mathcal{N}$ and
$$
(\mathcal{L}_{\tilde{\psi}_\theta} g)(x+)=(\mathcal{L}_{\tilde{\psi}_\theta} g)(x-)+\frac{\tilde{\psi}_{\theta}(x)}{S'(x)}\left[g'(x+)-g'(x-)\right]\leq 0
$$
for all $x\in (\hat{x}_\theta,b)\cap \mathcal{N}$. Consequently, we observe that our conditions guarantee that $(\mathcal{L}_{\tilde{\psi}_\theta} g)(x)$ does not decrease on $(x_0,\hat{x}_\theta)$, does not increase on $(\hat{x}_\theta,b)$ and, therefore, cannot change sign from positive to negative on $(a, \hat{x}_\theta)$. Denote now as $x' = \max\{x: x\in \mathcal{N}\}\vee \hat{x}_\theta$ and let $x > K > x'$. We then have
\begin{eqnarray*}
(\mathcal{L}_{\tilde{\psi}_\theta} g)(x) &=& (\mathcal{L}_{\tilde{\psi}_\theta} g)(K)+\int_{K}^x(\tilde{\mathcal{A}}_\theta g)(y)\tilde{\psi}_{\theta}(y)m'(y)dy\\
&=& (\mathcal{L}_{\tilde{\psi}_\theta} g)(K)+\frac{(\tilde{\mathcal{A}}_\theta g)(\xi)}{\theta}\left[\frac{\tilde{\psi}_{\theta}'(x)}{S'(x)}-\frac{\tilde{\psi}_{\theta}'(K)}{S'(K)}\right]\rightarrow -\infty
\end{eqnarray*}
as $x\rightarrow b$, since $(\tilde{\mathcal{A}}_\theta g)(x) < 0$ on $(\hat{x}_\theta,b)$ and $\lim_{x\rightarrow b}\tilde{\psi}_{\theta}'(x)/S'(x)=\infty$ for a natural boundary. Consequently, we observe that there exists a threshold $x_\theta^\ast\in[\hat{x}_\theta,b)$ at which $(\mathcal{L}_{\tilde{\psi}_\theta} g)(x)$ changes sign. Given the proven monotonicity of $(\mathcal{L}_{\tilde{\psi}_\theta} g)(x)$ on $(\hat{x}_\theta, b)$ then proves that this threshold is unique.
Noticing now that
$$
\frac{d}{dx}\left[\frac{g(x)}{\tilde{\psi}_{\theta}(x)}\right] = \frac{S'(x)}{\tilde{\psi}_{\theta}^2(x)}(\mathcal{L}_{\tilde{\psi}_\theta} g)(x)
$$
for $x\in \mathcal{I}\backslash \mathcal{N}$ demonstrates that $x_\theta^\ast=\argmax\left\{g(x)/\tilde{\psi}_{\theta}(x)\right\}$.

Let us now establish that the proposed value function dominates the value of any admissible $\mathbb{F}$-stopping strategy. Given the existence and uniqueness of $x_\theta^\ast$ we first observe that the proposed value function is nonnegative, continuous, continuously differentiable on $\mathcal{I}\backslash \mathcal{N}$, twice continuously differentiable on $\mathcal{I}\backslash \mathcal{N}$, satisfies the inequalities  $\tilde{V}_\theta'(x\pm) <\infty$ and $|\tilde{V}_\theta''(x\pm)| <\infty$ for all $x\in(x_\theta^\ast, b)\cap \mathcal{N}$, and dominates the exercise payoff $g(x)$. Moreover, $(\tilde{\mathcal{A}}_\theta \tilde{V}_\theta)(x) = 0$ for all $x\in (a, x_\theta^\ast)$. The assumed monotonicity of the function $(\mathcal{L}_{\tilde{\psi}_\theta} g)(x)$ and equation \eqref{operator} guarantee that  $(\tilde{\mathcal{A}}_\theta \tilde{V}_\theta)(x) \leq 0$ for all $x\in (x_\theta^\ast, b)\backslash \mathcal{N}$. It is now clear that our conditions guarantee that there exists a sequence $\{v_{j}\}_{j = 1}^{\infty}$ of mappings $v_{j}\in C^{2}(\mathcal{I})$ such that (cf. \citet{oksendal}, pp. 315--318)
\begin{itemize}
\item[] $v_{j}\rightarrow \tilde{V}_\theta$ uniformly on compact subsets of $\mathcal{I}$, as $j\rightarrow\infty$;
\item[] $\tilde{\mathcal{A}}_\theta v_{j}\rightarrow \tilde{\mathcal{A}}_\theta \tilde{V}_\theta$ uniformly on compact subsets of $\mathcal{I}\backslash \mathcal{N}$,
as $j\rightarrow\infty$;
\item[] $\{\tilde{\mathcal{A}}_\theta v_{j}\}_{j=1}^{\infty}$ is locally bounded on $\mathcal{I}$.
\end{itemize}
Let $\{A_N\}_{N\geq 1}$ be an increasing sequence of open subintervals of $\mathcal{I}$ satisfying the condition $A_N\uparrow \mathcal{I}$ as $N\uparrow\infty$. Applying the It{\^o}-Doeblin theorem to the mapping $(t,x)\mapsto e^{-\theta t}v_j(x)$ yields
$$
\mathbb{E}_{x}\left[e^{-\theta T_N}v_j(X_{T_N})\right] = v_j(x) + \mathbb{E}_{x}\int_0^{T_N}e^{-\theta s}(\tilde{\mathcal{A}}_\theta v_j)(\tilde{X}_s)ds,
$$
where $T_N=\tau\wedge \inf\{t\geq 0:X_t\not \in A_N\}\wedge N$ is a sequence of almost surely finite $\mathbb{F}$-stopping times converging to the arbitrary $\mathbb{F}$-stopping time $\tau$ as $N\uparrow\infty$. Reordering terms and applying Fatou's lemma now yields
\begin{eqnarray*}
\tilde{V}_\theta(x) &=& \lim_{j\rightarrow\infty}\mathbb{E}_{x}\left[e^{-\theta T_N}v_j(\tilde{X}_{T_N}) - \int_0^{T_N}e^{-\theta s}(\tilde{\mathcal{A}}_\theta v_j)(\tilde{X}_s)ds\right]\\
&\geq& \mathbb{E}_{x}\left[e^{-\theta T_N}\tilde{V}_\theta(\tilde{X}_{T_N}) - \int_0^{T_N}e^{-\theta s}(\tilde{\mathcal{A}}_\theta \tilde{V}_\theta)(\tilde{X}_s)ds\right]\\
&\geq& \mathbb{E}_{x}\left[e^{-\theta T_N}g(\tilde{X}_{T_N})\right].
\end{eqnarray*}
Letting $N\uparrow\infty$ and applying Fatou's lemma again then demonstrates that the proposes value function satisfies the inequality
$$
\tilde{V}_\theta(x) \geq \mathbb{E}_{x}\left[e^{-\theta \tau}g(\tilde{X}_{\tau})\right]
$$
for any any admissible stopping time. Hence, it dominates the value of the optimal policy. However, since the proposed value is attained by the Markov time $\tilde{\tau}_{x_\theta^\ast}=\inf\{t\geq 0: \tilde{X}_t\geq x_\theta^\ast\}$ belonging into the larger class of $\mathbb{F}$-stopping times, we finally notice that the proposed value actually constitutes the value in \eqref{associatedprob}.
\end{proof}

Lemma \ref{l1} expresses a set of sufficiency conditions under which the optimal stopping strategy of the associated stopping
problems constitutes a standard threshold policy. It is clear that by imposing more smoothness assumptions on the exercise payoff result
in more easily verifiable sufficiency conditions.

\section{The Representation Theorem}

Our objective is to demonstrate that the value function
of the stopping problem \eqref{stoppingproblem} can be expressed in the familiar form
\begin{eqnarray*}
V_\lambda(x)=\psi_\lambda(x) \sup_{y \geq x}\left\{ \frac{g(y)}{\psi_\lambda(y)}\right\}
\end{eqnarray*}
in the spectrally negative jump diffusion setting as well. Our main result stating a set of sufficient conditions under which the standard representation is satisfied is established in the following.

\begin{theorem}\label{mainthm}
Assume that condition ({\bf A1}) is satisfied and that the assumptions of Lemma \ref{l1} are met for $\theta\in[r,r+\lambda]$ with the additional
condition that $(\hat{x}_{r+\lambda},\hat{x}_{r})\cap\mathcal{N}=\emptyset$. Assume also that at least one of the following conditions hold:
\begin{itemize}
  \item[(i)] there is a unique $\hat{x}\in [\hat{x}_{r+\lambda},\hat{x}_{r}]$ such that $(\mathcal{G}_rg)(x) \gtreqqless 0$ for all $x \lesseqqgtr \hat{x}, x\not\in \mathcal{N}$,
  \item[(ii)] $\mathfrak{m}(\{z\in \mathcal{S}:x+\gamma(x,z)< x_0\}) > 0$ for all $x\in [x_{r+\lambda}^\ast,x_{r}^\ast]$.
\end{itemize}
 Then, the value of the optimal stopping strategy $\tau^\ast =\inf\{t\geq 0: X_t \geq x^\ast\}$ reads as
$$
V_\lambda(x)=\psi_\lambda(x) \sup_{y \geq
x}\left\{\frac{g(y)}{\psi_\lambda(y)}\right\}=\begin{cases}
g(x) &x\geq x^\ast\\
g(x^\ast)\frac{\psi_\lambda(x)}{\psi_\lambda(x^\ast)} & x < x^\ast,
\end{cases}
$$
where $x^\ast=\argmax\{g(x)/\psi_\lambda(x)\}$ constitutes the optimal exercise threshold.
\end{theorem}
\begin{proof}
We first notice by combining the results of Lemma \ref{domrexcessive} and Lemma \ref{l1} that $\tilde{V}_{r+\lambda}(x)\leq V_\lambda(x) \leq \tilde{V}_{r}(x)$, where the
values $\tilde{V}_{r+\lambda}(x)$ and $\tilde{V}_{r}(x)$ can be expressed as in \eqref{value}. These inequalities imply that $V_\lambda(x)=g(x)$ for all $x\geq x_{r}^\ast$ and $V_\lambda(x)>g(x)$ for all $x\leq x_{r+\lambda}^\ast$. Given these observations, we now plan to establish that our assumptions are sufficient for guaranteeing the existence of a unique exercise threshold $x^\ast\in [x_{r+\lambda}^\ast, x_{r}^\ast]$ maximizing the ratio $g(x)/\psi_\lambda(x)$. As is clear from the proof of Corollary \ref{dominationcorollary}, the ratio $g(x)/\psi_\lambda(x)$ is decreasing on the set where the ratio $g(x)/\tilde{\psi}_r(x)$ is decreasing and increasing on the set where
the ratio $g(x)/\tilde{\psi}_{r+\lambda}(x)$ is increasing. Consequently, we observe that $g(x)/\psi_\lambda(x)$ has at least one maximum point on $[x_{r+\lambda}^\ast,x_{r}^\ast]$. Denote now as $\Lambda$ the set of maximum points of $g(x)/\psi_\lambda(x)$ and let $x' = \max\{x: x\in \Lambda\}$ denote the maximal element of that set. Consider now the function
$$
\check{V}(x)=\begin{cases}
g(x) &x\geq x'\\
g(x')\frac{\psi_\lambda(x)}{\psi_\lambda(x')} & x < x'.
\end{cases}
$$
Since $$\check{V}(x) = \mathbb{E}_x\left[e^{-r\tau_{x'}}g(X_{\tau_{x'}})\right],$$ where $\tau_{x'} = \inf\{t\geq 0:X_t\geq x'\}$ is an admissible stopping strategy, we notice that $\check{V}(x)\leq V_\lambda(x)$. On the other hand, it is also clear that $\check{V}(x)$ is nonnegative, continuous, dominates the exercise payoff $g(x)$, and belongs to $C^1(\mathcal{I}\backslash \mathcal{N})\cap C^2(\mathcal{I}\backslash \mathcal{N})$. Moreover, since $(\mathcal{G}_r\check{V})(x)=0$ on $(0,x')$ it is sufficient to analyze the behavior of $(\mathcal{G}_r\check{V})(x)$ on $(x',b)\backslash \mathcal{N}$. Since $V_\lambda(x) = \tilde{V}_r(x)=g(x)$ on $[x_r^\ast,b)$ by Lemma \ref{domrexcessive} and
$$
(\mathcal{G}_rg)(x) = (\tilde{\mathcal{A}}_rg)(x)+\lambda\int_{\mathcal{S}}(g(x+\gamma(x,z))-g(x))\mathfrak{m}(dz) \leq (\tilde{\mathcal{A}}_rg)(x) \leq 0
$$
for all $x\in[\hat{x}_r,b)\backslash \mathcal{N}$ we observe that if $x'\geq \hat{x}_r$ then $(\mathcal{G}_r\check{V})(x)\leq 0$ on $(x',b)\backslash \mathcal{N}$ and we are done. Assume, therefore, that $x' <\hat{x}_r$ and define the continuous mapping $\check{g}(x) = g(x)\mathbf{1}_{[x_0,b)}(x)+(g(x)+\delta(x-x_0)) \mathbf{1}_{(a,x_0)}(x)$, where $\delta>0$ is a known positive constant. It is clear that $\check{g}(x) \leq g(x)$ for all $x\in (a,b)$ and $\max(\check{g}(x),0)=\max(g(x),0)$. Consequently, the value of the optimal stopping problem defined with respect to the exercise payoff $\check{g}(x)$ coincides with the value $V_\lambda(x)$ (since $V_\lambda(x)$ is $r$-excessive and dominates $g^+(x)$). Assumption $\mathfrak{m}(\{z\in \mathcal{S}:x+\gamma(x,z)< x_0\}) > 0$ for all $x\in [x_{r+\lambda}^\ast,x_{r}^\ast]$ now implies that we can always choose the parameter $\delta$ so that
$$
\sup_{x\in [x_{r+\lambda}^\ast,x_{r}^\ast]}\int_{\mathcal{S}}\check{g}(x+\gamma(x,z))\mathfrak{m}(dz) < 0.
$$
This inequality guarantees that for all $x\in [x',\hat{x}_r)\backslash \mathcal{N}$ we have
$$
0 \geq (\tilde{\mathcal{A}}_{r+\lambda}\check{g})(x)\geq (\tilde{\mathcal{A}}_{r+\lambda}\check{g})(x) +\lambda \int_{\mathcal{S}}\check{g}(x+\gamma(x,z))\mathfrak{m}(dz) = (\mathcal{G}_r\check{g})(x).
$$
Combining these inequalities with the sufficient smoothness of $\check{V}$ and the technique based on a sequence of smooth functions converging uniformly to the value $\check{V}$ applied in the proof of Lemma \ref{l1} then imply that $x^\ast=x'$ and $\check{V}(x)\geq V_\lambda(x)$ (since $V$ is the smallest $r$-excessive majorant of $g$ for $X$).
\end{proof}

Theorem \ref{mainthm} demonstrate that the sufficient conditions guaranteeing the existence of an optimal threshold in the continuous diffusion setting
are sufficient in the jump diffusion setting as well provided that the support of the jump size distribution is sufficiently extensive. An important implication of  Theorem \ref{mainthm} is summarized in the following:
\begin{corollary}\label{discountcor}
Assume that the conditions of Theorem \ref{mainthm} are satisfied. Then, there exists a jump risk adjusted discount rate $\theta^\ast\in [r, r+\lambda]$ so that
$x_{\theta^\ast}^\ast=x^\ast$, that is, so that the optimal exercise threshold in the absence of jumps coincides with the one in the presence of jumps.
\end{corollary}
\begin{proof}
As we know from Theorem \ref{mainthm}, $x^\ast\in [x_{r+\lambda}^\ast,x_r^\ast]$. However, since $x_\theta^\ast$ is continuous and monotonically decreasing as a function of the prevailing discount rate $\theta$, we find that $x_\theta^\ast=x^\ast$ has a unique root in $[r, r+\lambda]$.
\end{proof}

According to Corollary \ref{discountcor} the optimal exercise boundary $x^\ast$ can be attained in the continuous diffusion setting by adjusting the discount rate appropriately for the jump risk. It is worth noticing that since $x^\ast > x_{r+\lambda}^\ast$ the same conclusion can be drawn by adjusting the growth rate $\tilde{\alpha}(x)$ appropriately as well. We will illustrate this observation in our explicit examples.

\section{Comparative Statics}

In this section our main objective is to consider comparative static
properties of the value function and the optimal policy and,
especially, to analyze the impact of increased volatility on these
factors. To this end, we consider two jump diffusions of the form
\eqref{jumpdiffusion}, $X$ and $\hat{X}$, which are otherwise
identical but have different volatilities,
$\sigma(x)>\hat{\sigma}(x)$. In accordance with this notation, we
denote the values of the associated optimal stopping problems by $V_\lambda$
and $\hat{V}_\lambda$, the associated integro-differential operators as
$\mathcal{G}_r$ and $\mathcal{\hat{G}}_r$, and the associated
increasing fundamental solutions (given that assumption ({\bf A1}) is
satisfied) as $\psi_\lambda$ and $\hat{\psi}_\lambda$, respectively. Our first
result emphasizing the role of these fundamental solutions is now
summarized in the next theorem.
\begin{theorem}\label{incvol}
Assume that the increasing fundamental solution $\psi_\lambda(x)$ is convex.
Then
$$
\frac{\hat{\psi}_\lambda(x)}{\hat{\psi}_\lambda(y)}\leq\frac{\psi_\lambda(x)}{\psi_\lambda(y)}\quad \textrm{and } \quad \frac{\hat{\psi}_\lambda'(x)}{\hat{\psi}_\lambda(x)}\geq\frac{\psi_\lambda'(x)}{\psi_\lambda(x)}
$$
for all $x\leq y$.  Moreover, if the conditions of
Theorem \ref{mainthm} are satisfied, then $V_\lambda(x)\geq
\hat{V}_\lambda(x)$ and, therefore,
$$
\hat{C}=\{x\in \mathcal{I}: \hat{V}_\lambda(x) > g(x)\}\subseteq \{x\in \mathcal{I}: V_\lambda(x) >
g(x)\}=C.
$$
If the increasing fundamental solution $\hat{\psi}_\lambda(x)$ is concave,
then the inequalities and inclusions stated above are reversed.
\end{theorem}
\begin{proof}
Applying Dynkin's formula to
$(t,x)\mapsto e^{-rt}\psi_\lambda(x)$ yields
\begin{eqnarray*}
\mathbb{E}_{x}[e^{-r \hat{\tau}_{(a,y)}}
\psi_\lambda(\hat{X}_{\hat{\tau}_{(a,y)}})]=\psi_\lambda(x)+\mathbb{E}_{x}\int_{0}^{\hat{\tau}_{(a,y)}}e^{-rt}
(\mathcal{\hat{G}}_{r}\psi_\lambda)(\hat{X}_{t})dt,
\end{eqnarray*}
where $\hat{\tau}_{(a,y)}=\inf\{t\geq 0:\hat{X}_t\geq y\}$. Since
$\hat{X}_{\hat{\tau}_{(a,y)}} = y$ a.s. and
$(\mathcal{\hat{G}}_{r}\psi_\lambda)(x) =
\left((\mathcal{\hat{G}}_{r}-\mathcal{G}_r+\mathcal{G}_r)\psi_\lambda\right)(x)
=
\left((\mathcal{\hat{G}}_{r}-\mathcal{G}_r)\psi_\lambda\right)(x)=\frac{1}{2}(\hat{\sigma}^2(x)-\sigma^2(x))\psi_\lambda''(x)\leq
0$ by the $X$-harmonicity and convexity of $\psi_\lambda(x)$, we find
that
$$
\mathbb{E}_{x}[e^{-r \hat{\tau}_{(a,y)}}]\psi_\lambda(y) =
\frac{\hat{\psi}_\lambda(x)}{\hat{\psi}_\lambda(y)}\psi_\lambda(y) \leq\psi_\lambda(x).
$$
The second inequality can be established as in Corollary \ref{dominationcorollary}.
Establishing the reverse conclusions in case the
fundamental solution  $\psi_\lambda(x)$ is concave is completely analogous.
\end{proof}

Theorem \ref{incvol} extends previous findings based on continuous
diffusions to the present setting as well and states a set of
conditions in terms of the convexity (concavity) of the fundamental
solution $\psi_\lambda(x)$ under which increased volatility unambiguously
decelerates (accelerates) rational exercise by expanding (shrinking)
the continuation region where waiting is optimal. As is clear from
this observation, the sign of the relationship between increased
volatility and the optimal stopping policy is a process-specific
property that as such does not depend on the precise form of the
exercise payoff as long as the supremum at which the expected
present value of the payoff is maximized exists and constitutes the
optimal stopping rule.

It is worth noticing that the proof of our Theorem \ref{incvol}
indicates that the analysis of the impact of increased volatility on
the optimal policy and its value reduces to the comparison of the
$r$-superharmonic mappings characterized by the integro-differential
operators $\mathcal{G}_r$ and $\hat{\mathcal{G}}_r$. Since
$(\mathcal{\hat{G}}_ru)(x) \leq (\mathcal{G}_ru)(x)$ for any
sufficiently smooth convex function $u:\mathcal{I}\mapsto \mathbb{R}_+$ and
$(\mathcal{\hat{G}}_rv)(x) \geq (\mathcal{G}_rv)(x)$ for any
sufficiently smooth concave function $v:\mathcal{I}\mapsto \mathbb{R}_+$, we
find that the findings of our Theorem \ref{incvol} generate a
natural ordering for the convex (concave) solutions of the
variational inequalities $\max\{(\mathcal{\hat{G}}_ru)(x),
g(x)-u(x)\}=0$ and $\max\{(\mathcal{G}_ru)(x),g(x)-u(x)\}=0$.

Having characterized the impact of increased volatility on the
optimal policy and its value, it is naturally of interest to analyze
how the jump-intensity $\lambda$ measuring the rate at which the
downside risk is realized affects these factors. Along the lines of
our previous notation, we now consider two jump diffusions of the
form \eqref{jumpdiffusion}, $X$ and $\hat{X}$, which are otherwise
identical but are subject to different jump intensities, $\lambda
> \hat{\lambda}$. In line with this notation, we denote the associated integro-differential operators by
$\mathcal{G}_r$ and $\mathcal{\hat{G}}_r$, respectively. Our main characterization on the impact of increased
jump intensity on the value and the optimal policy is now summarized
in our next theorem.
\begin{theorem}\label{incintensity}
Assume that the increasing fundamental solution $\psi_\lambda(x)$ is
convex. Then
$$
\frac{\psi_{\hat{\lambda}}(x)}{\psi_{\hat{\lambda}}(y)}\leq\frac{\psi_\lambda(x)}{\psi_\lambda(y)}\quad \textrm{and }\quad
\frac{\psi_{\hat{\lambda}}'(x)}{\psi_{\hat{\lambda}}(x)}\geq\frac{\psi_\lambda'(x)}{\psi_\lambda(x)}
$$
for all $x\leq y$. Moreover, if the conditions of
Theorem \ref{mainthm} are satisfied, then
$V_\lambda(x)\geq V_{\hat{\lambda}}(x)$ and, therefore,
$$
C_{\hat{\lambda}}=\{x\in \mathcal{I}: V_{\hat{\lambda}}(x) > g(x)\}\subseteq
\{x\in \mathcal{I}: V_\lambda(x)
> g(x)\}=C_\lambda.
$$
If the increasing fundamental solution $\psi_{\hat{\lambda}}(x)$ is
concave, then the inequalities and inclusions stated above are
reversed.
\end{theorem}
\begin{proof}
The assumed convexity of the increasing fundamental solution
$\psi_\lambda(x)$ implies that $\psi_\lambda(x + \gamma(x,z))\geq
\psi_\lambda(x) + \psi_\lambda'(x)\gamma(x,z)$ for any $z\in
\mathcal{S}$ and, therefore, that
$$
\int_{\mathcal{S}}\{\psi_\lambda(x + \gamma(x,z))- \psi_\lambda(x) -
\psi_\lambda'(x)\gamma(x,z)\}\mathfrak{m}(dz)>0.
$$
Consequently, we observe that
$$
(\mathcal{\hat{G}}_r\psi_\lambda)(x) = (\hat{\lambda} -
\lambda)\int_{\mathcal{S}}\{\psi_\lambda(x + \gamma(x,z))-
\psi_\lambda(x) - \psi_\lambda'(x)\gamma(x,z)\}\mathfrak{m}(dz) < 0
$$
for all $x\in \mathcal{I}$. Applying now Dynkin's theorem to $\psi_\lambda(x)$
then finally proves that
$\psi_{\hat{\lambda}}(x)/\psi_{\hat{\lambda}}(y)\leq
\psi_\lambda(x)/\psi_\lambda(y)$ for $x\leq y$. Establishing the rest of the
alleged results is completely analogous with the proof of Theorem \ref{impactofvol}.
\end{proof}

Theorem \ref{incintensity} characterizes how the direction of the
impact of increased jump-intensity $\lambda$ on the optimal stopping
policy and its value can be unambiguously determined when the
fundamental solution is convex (concave). Along the lines of our
findings on the impact of increased volatility, we observe that
higher jump-intensity also slows down (speeds up) rational exercise
by expanding (shrinking) the continuation region when $\psi(x)$ is
convex (concave). This result is economically important, since it
essentially states that if the value is convex on the continuation
region where exercising is suboptimal, then the combined impact of
downside risk and systematic market risk on the exercise incentives
of rational investors is unambiguously negative.

We  next state a set of sufficient conditions for the convexity of the value (and $\psi_\lambda(x)$) when
the underlying process is the slightly less general
\begin{eqnarray*}
X_{t}= \int_0^t\alpha(X_{s}) ds + \int_0^t\sigma X_{s} d W_{s} +
\int_0^t\int_{\mathcal{S}} \gamma(z)X_{s} \tilde{N}(dz,ds).
\end{eqnarray*}
In this setting we can state
the following sufficient conditions for the convexity of the value function.
\begin{theorem}\label{impactofvol}
Suppose that $g$ and $\alpha$ are convex functions, that $\alpha(x)$
has a locally Lipschitz continuous derivative, and that
$r x-\alpha(x)$ is increasing. Then the value function of the stopping problem is
convex.
\end{theorem}
\begin{proof}
We denote $Y_{t}^{1}:=\frac{\partial X_{t}}{\partial x}$. By virtue
of Theorem V.40 of \citet{protter}, we can differentiate the flow
$X_{t}=X_{t}^{x}$ with respect to the initial state $x$ to obtain
\begin{eqnarray*}
Y_{t}^1= \int_0^t\alpha^{\prime}(X_{s}^x)Y_s^1 ds + \int_0^t\sigma
Y_s^1 d W_{s} + \int_0^t\int_{\mathcal{S}} \gamma(z)Y_s^1
\tilde{N}(dz,ds),
\end{eqnarray*}
which implies that
\begin{eqnarray*}
Y_t^1 = \exp\left(\int_0^t \alpha^{\prime}(X_{s}^x)ds\right)\mathcal{E}_{t}\geq
0,
\end{eqnarray*}
where
\begin{eqnarray*}
\mathcal{E}_{t}=\exp\left(\sigma W_t-\frac{1}{2}\sigma^2
t+\int_{0}^{t}\int_{\mathcal{S}}\ln (1+\gamma(z))N(ds,dz)-\lambda
\overline{\gamma} t\right)
\end{eqnarray*}
is a positive exponential martingale independent of $x$ and
$$
\overline{\gamma}:=\int_{\mathcal{S}}\gamma(z)\mathfrak{m}(dz).
$$
Thus differentiating the mapping
$$
Q(t,x) := \mathbb{E}\left[e^{-rt}g(X_t^x)\right]
$$
with respect to $x$ yields
$$
Q_x(t,x) =
\mathbb{E}\left[\exp\left(-\int_0^t(r-\alpha'(X_{s}^x))ds\right)g'(X_t^x)M_t\right]\geq
0,
$$
which as a function of $x$ is increasing, being under our
assumptions the product of two non-negative and monotonically
increasing functions. Thus $Q(t,x)$ is an increasing and convex
function of $x$. Consequently, all elements of the increasing
sequence $\{V_k(x)\}_{k\in \mathbb{N}}$ defined by
$$
V_0(x) = \sup_{t\geq 0}\mathbb{E}\left[e^{-rt}g(X_t^x)\right]
$$
$$
V_{k+1}(x) = \sup_{t\geq 0}\mathbb{E}\left[e^{-rt}V_k(X_t^x)\right]
$$
are increasing and convex. Furthermore, $V_k(x)\uparrow V_\lambda(x)$. If
$\alpha\in[0,1]$ and $x,y\in \mathcal{I}$, then
\begin{eqnarray*}
\alpha V_\lambda(x) + (1-\alpha)V_\lambda(y) &\geq& \alpha V_k(x) +
(1-\alpha)V_k(y)\\ &\geq&  V_k(\alpha x+(1-\alpha)y)
\end{eqnarray*}
for all $k$. By monotone convergence
$$
\alpha V_\lambda(x) + (1-\alpha)V_\lambda(y)\geq \lim_{k\rightarrow\infty}V_k(\alpha
x+(1-\alpha)y) = V_\lambda(\alpha x+(1-\alpha)y),
$$
which implies the convexity of the value $V_\lambda$.
\end{proof}

Theorem \ref{impactofvol} states a set of conditions under which the
sign of the relationship between increased volatility and the value
of the considered optimal stopping problem is unambiguously
positive. It is worth noticing that along the lines of the findings
by \citet{alvarez4} the monotonicity of
$\mu(x)-rx$ is the key factor determining how higher volatility
affects the optimal policy. The reason for this observation is
naturally the fact that our evaluations are based on the compensated
compound Poisson process (which is a martingale). If this were not
the case, then the local expected behavior of the underlying jump
process would naturally have a constant effect on the monotonicity
requirement stated in Theorem \ref{impactofvol}.

\section{Certainty Equivalent Valuation}

Having developed a set of sufficient conditions under which the the optimal exercise strategy of the considered stopping problem constitutes  a standard single threshold policy, we now proceed in our analysis and investigate along the lines of the study \citet{alvarez5} the following question: {\em Can the value and optimal exercise threshold be expressed as a solution to an associated deterministic timing problem adjusted to the risk generated by the driving Lévy process?} To address this question we now introduce an associated deterministic process labeled as $\hat{X}_t$ evolving according to the dynamics characterized by the ordinary differential equation
\begin{eqnarray}
\hat{X}_t' = \hat{\mu}_\lambda(\hat{X}_t),\quad X_0=x,\label{det1}
\end{eqnarray}
where $\hat{\mu}_\lambda(x)$ is a continuous and nonnegative function specified below. Having presented the associated deterministic dynamics, we now introduce the associated valuation
\begin{eqnarray}
\hat{V}(x) = \sup_{t\geq 0} e^{-rt}g(\hat{X}_t).
\end{eqnarray}
Our main result on certainty equivalent valuation is now summarized in the following.
\begin{theorem}\label{certaintyequiv}
Assume that the conditions of Theorem \ref{mainthm} are satisfied and define the risk adjusted growth rate as
\begin{eqnarray}
\hat{\mu}_\lambda(x)= \frac{r\psi_\lambda(x)}{\psi_\lambda'(x)} =
\alpha(x) + \frac{1}{2}\sigma^2(x)
\frac{\psi_\lambda''(x)}{\psi_\lambda'(x)}.\label{certeqadj}
\end{eqnarray}
Then, $V_\lambda(x) \equiv \hat{V}(x)$. Moreover, if $\psi_\lambda(x)$ is convex (concave), then increased volatility increases (decreases) and increased jump intensity increases (decreases) the risk adjusted growth rate.
\end{theorem}
\begin{proof}
Assume that \eqref{certeqadj} is satisfied and consider the mapping $(t,x)\mapsto e^{-rt}g(\hat{X}_t)$. Standard differentiation yields
$$
\frac{d}{dt}\left[e^{-rt}g(\hat{X}_t)\right] = \frac{re^{-rt}}{\psi_\lambda'(\hat{X}_t)}\left[g'(\hat{X}_t)\psi_\lambda(\hat{X}_t)-g(\hat{X}_t)\psi_\lambda'(\hat{X}_t)\right].
$$
As was demonstrated in Theorem \ref{mainthm}, there is a unique threshold $x^\ast = \argmax\{g(x)/\psi_\lambda(x)\}$ for which $g'(x)\psi_\lambda(x)\gtreqqless g(x)\psi_\lambda'(x)$ when $x \lesseqqgtr x^\ast$. Since $\hat{\mu}_\lambda(x)= \frac{r\psi_\lambda(x)}{\psi_\lambda'(x)} > 0$ for all $x\in \mathcal{I}$, we notice that $$t^\ast = \inf\{t\geq 0: \hat{X}_t\geq x^\ast\}=\frac{1}{r}\ln\left(\max\left(\frac{\psi_\lambda(x^\ast)}{\psi_\lambda(x)},1\right)\right)$$ is the optimal stopping time and that $V_\lambda(x) \equiv \hat{V}(x)$. The positivity of the sensitivity of the risk adjusted growth rate with respect to changes in the jump intensity follows from Theorem \ref{incintensity} and  with respect to changes in volatility from Theorem \ref{incvol}.
\end{proof}

Theorem \ref{certaintyequiv} extends the findings of \citet{alvarez5} and shows that the value and optimal exercise boundary of the optimal stopping problem \eqref{stoppingproblem} of a discontinuous jump diffusion coincide with the value and stopping boundary of a stopping problem of a continuous and deterministic process. According to Theorem \ref{certaintyequiv}, this identity can be attained by adjusting appropriately the growth rate of the deterministic dynamics to the uncertainty generated by the driving Lévy process. Since this adjustment can be made for the continuous diffusion model as well, \eqref{certeqadj} can be applied in decomposing the risk adjusted growth rate into two parts capturing the uncertainty of the driving stochastic processes. More precisely,
in the absence of jumps (i.e. when $\lambda\equiv 0$) the risk adjusted growth rate reads as
$$
\hat{\mu}_0(x)= \frac{r\psi_0(x)}{\psi_0'(x)}.
$$
Consequently, if the increasing fundamental solution $\psi_\lambda(x)$ is convex, then Theorem \ref{incintensity} implies that
$$
\hat{\mu}_\lambda(x)-\hat{\mu}_0(x)= r\left(\frac{\psi_\lambda(x)}{\psi_\lambda'(x)}-\frac{\psi_0(x)}{\psi_0'(x)}\right) >0.
$$
The opposite conclusion is valid in case the increasing fundamental solution is concave.

It is also worth emphasizing that the certainty equivalent valuation formula presented in Theorem \ref{certaintyequiv} can be extended within the spectrally negative jump diffusion setting to cases where the existence of a smooth solution $\psi(x)$ satisfying ({\bf A1}) is not necessarily straightforward to establish. More precisely, if the lower boundary cannot be attained in finite time then Theorem 8.1 of \citet{kyprianou} implies that
$$
\mathbb{E}_x\left[e^{-r\tau_y}g(X_{\tau_y}); \tau_y<\infty\right]=\begin{cases}
g(x) &x\geq y\\
g(y) \frac{F^{(r)}(x)}{F^{(r)}(y)} &x < y
\end{cases}
$$
where $F^{(r)}(x)$ denotes the $r$-scale function associated with $X_t$ and $\tau_y=\inf\{t\geq 0:X_t\geq y\}$. Since the paths of the underlying process are of unbounded variation, we know that
$F^{(r)}\in C^1(\mathcal{I})$. Therefore, choosing
$$
\hat{\mu}_\lambda(x) = r\frac{F^{(r)}(x)}{{F^{(r)}}'(x)}
$$
as the risk adjusted growth rate of the deterministic process $\hat{X}_t$ then shows that
$$
\mathbb{E}_x\left[e^{-r\tau_y}g(X_{\tau_y})\right] = e^{-rT_y}g(\hat{X}_{T_y}),
$$
where $T_y = \inf\{t\geq 0: \hat{X}_t\geq y\}$. Consequently, {\em if the optimal stopping strategy is known to constitute a standard
single exercise threshold policy, then the value of the optimal policy can always be expressed in terms of a risk adjusted deterministic
valuation}.

\section{Explicit Illustrations}

In this section our objective is to illustrate our main findings
within explicitly parametrized examples based on different
descriptions for the underlying stochastic dynamics. As usually, we
illustrate our findings for the arithmetic L\'{e}vy process
and the geometric L\'{e}vy process since in those cases the
representation obtained in the analysis of our previous sections is
valid.

\subsection{Arithmetic Stochastic Dynamics}

Consider first the arithmetic case
\begin{eqnarray*}
dX_{t}=\mu dt + \sigma dW_{t} - \int_{\mathcal{S}} \gamma z
\tilde{N}(dt,dz),\quad X_0=x
\end{eqnarray*}
where $\mu,\sigma,\gamma\in \mathbb{R}_+$ and $\mathcal{I}=\mathbb{R}$. For simplicity, we assume that $\mathcal{S}=\mathbb{R}_+$. In this case the associated integro-differential
equation
\begin{eqnarray*}
\frac{1}{2} \sigma^{2} \psi''(x)+(\mu+\gamma \lambda
\overline{m}) \psi'(x)-(r+\lambda)\psi(x)+\lambda
\int_{0}^{\infty}\psi(x-\gamma z)\mathfrak{m}(dz)=0
\end{eqnarray*}
has an increasing solution $\psi(x)=e^{k_{1}x}$ where $k_{1}>0$ solves
\begin{eqnarray*}
\frac{1}{2} \sigma^{2} k^2+(\mu+\gamma \lambda
\overline{m})k+\lambda \int_{0}^{\infty}e^{-\gamma z k}
\mathfrak{m}(dz)-(r+\lambda)=0.
\end{eqnarray*}
This equation also implies that if $\mu>0$ then $\lim_{r\downarrow 0}k_1 = 0$.
Consequently, we observe that in that case
$$
\mathbb{P}_x[\tau_y < \infty] = \lim_{r\downarrow 0}\mathbb{E}_x\left[e^{-r\tau_y};\tau_y < \infty\right] = \lim_{r\downarrow 0}e^{k_1(x-y)} = 1.
$$
It is now clear that if the sufficiency conditions of Theorem \ref{mainthm} are satisfied then
the value of the optimal stopping policy can be represented as
\begin{eqnarray} \label{ex1rep}
V_\lambda(x)=e^{k_{1}x} \sup_{y \geq x}\left\{e^{-k_{1}y}g(y)\right\}=
\begin{cases}
g(x), & x \geq x^{\ast},\\
g(x^{\ast})e^{k_{1}(x-x^{\ast})}, & x < x^{\ast},
 \end{cases}
\end{eqnarray}
where $x^{\ast} =\argmax\{e^{-k_{1}x}g(x)\}$ satisfies for a
differentiable $g$ the ordinary first order condition $D_{x}[\ln g(x)]=k_{1}$. It is
also worth pointing out that in accordance with the findings of our
Lemma \ref{domination} we now find that the root $k_1\in
(\tilde{k}_{r}, \tilde{k}_{r+\lambda})$, where
$$
\tilde{k}_\theta = -\frac{\mu +\gamma \lambda \bar{m}}{\sigma^2} +
\sqrt{\left(\frac{\mu +\gamma \lambda
\bar{m}}{\sigma^2}\right)^2+\frac{2\theta}{\sigma^2}}
$$
denotes the positive root of the characteristic equation $\sigma^{2}
k^2+2(\mu+\gamma \lambda \overline{m})k=2\theta$. Consequently, we
observe that in the present setting
$$
e^{\tilde{k}_{r+\lambda}x} \sup_{y \geq
x}\left\{e^{-\tilde{k}_{r+\lambda}y}g(y)\right\}\leq e^{k_{1}x} \sup_{y \geq
x}\left\{e^{-k_{1}y}g(y)\right\}\leq e^{\tilde{k}_{r}x} \sup_{y \geq
x}\left\{e^{-\tilde{k}_{r}y}g(y)\right\}
$$
provided that the maximum exists. The jump risk adjusted discount rate $\theta^\ast$ defined in Corollary \ref{discountcor} for which $x^\ast_{\theta^\ast}=x^\ast$ reads now as
$$
\theta^\ast=r+\lambda\int_{0}^{\infty}(1-e^{-\gamma z k_1})
\mathfrak{m}(dz).
$$
Alternatively, letting the drift coefficient of the associated continuous diffusion $\tilde{X}_t$
to be
$$
\tilde{\mu}=\mu+\gamma \lambda
\overline{m} + \frac{\lambda}{k_1}\int_{0}^{\infty}e^{-\gamma z k_1}
\mathfrak{m}(dz)
$$
results into the equality $x_{r+\lambda}^\ast=x^\ast$.
Finally, as demonstrated in Theorem \ref{certaintyequiv},
choosing  $\hat{\mu}_\lambda(x)=r/k_1$ as the risk-adjusted growth rate implies that
$\hat{V}(x)=V_\lambda(x)$.

It is worth noticing that according to our general results the
strict convexity of the increasing fundamental solution $e^{k_{1}x}$
implies that increased volatility $\sigma$ as well as higher
jump-intensity $\lambda$ increases the value of the optimal stopping
policy and raises the optimal boundary at which the underlying
jump-diffusion should be stopped. An analogous conclusion is naturally valid for the risk adjusted
growth rate $\hat{\mu}_\lambda(x)=r/k_1$ as well. Moreover, since
$$
\lim_{\lambda\downarrow 0}k_1 = -\frac{\mu}{\sigma^2}+\sqrt{\frac{\mu^2}{\sigma^4}+\frac{2r}{\sigma^2}}
$$
we notice that
$$
\hat{\mu}_\lambda(x)-\hat{\mu}_0(x)=r\left[\frac{1}{k_1}-\frac{1}{-\frac{\mu}{\sigma^2}+\sqrt{\frac{\mu^2}{\sigma^4}+\frac{2r}{\sigma^2}}}\right].
$$
We illustrate the risk adjusted growth rate for $\Gamma(a,b)$-distributed jumps in Table \ref{table1} under the assumptions that $\mu=0.04,r=0.05,\gamma=1,a=1$ and $b=1$.
As these numerical values indicate, jump risk has a nontrivial effect on the required risk adjustment in the present setting even when the intensity of the jump process is relatively low.
\begin{table}[!ht]
  \centering
  \begin{tabular}{|c|c|c|c|c|c|}
  \hline
  $\sigma$ & 0.05&0.1&0.15&0.2&0.25 \\
  \hline
  $\hat{\mu}_0-\mu$ & 0.15&0.55&1.10&1.74&2.43\\
  $\hat{\mu}_{0.1}-\mu$ & 3.94&4.12&4.40&4.77&5.21 \\
  $\hat{\mu}_{0.2}-\mu$ &6.51&6.63&6.83&7.11&7.45\\
 \hline
\end{tabular}
  \caption{Risk Adjusted Growth Rates in Percentage Terms}\label{table1}
\end{table}

As a numerical illustration, consider the  \emph{capped option}
reward function
\begin{eqnarray*}
g(x)=(\min(K,x)-I)^+,
\end{eqnarray*}
where we assume $K>I>0$ (cf. \citet{alvarez3}). It is clear that the conditions of Theorem \ref{mainthm}
are satisfied and that
$$
x^\ast = \begin{cases}
I+\frac{1}{k_1} &k_1\geq (K-I)^{-1}\\
K& k_1< (K-I)^{-1}.
\end{cases}
$$
Hence the value of the
optimal stopping problem reads as
\begin{eqnarray*}
V_\lambda(x)=e^{k_{1}x} \sup_{y \geq x}\left\{e^{-k_{1}y}g(y)\right\} =
\begin{cases}
K-I, & x \geq K\\
e^{k_{1}(x-K)} (K-I), & x < K
\end{cases}
\end{eqnarray*}
when $k_1< (K-I)^{-1}$ and as
\begin{eqnarray*}
V_\lambda(x)=e^{k_{1}x} \sup_{y \geq x}\left\{e^{-k_{1}y}g(y)\right\} =
\begin{cases}
K-I, & x > K\\
x-I, & I+\frac{1}{k_1}\leq x < K\\
\frac{1}{k_1}e^{k_1(x-I-1/k_1)} & x < I+\frac{1}{k_1}
\end{cases}
\end{eqnarray*}
when $k_1\geq (K-I)^{-1}$. Especially, when $k_1< (K-I)^{-1}$ we find that
\begin{eqnarray*}
\lim_{x \rightarrow K-}V^{\prime}(x)=k_{1}(K-I) > 0 =
\lim_{x \rightarrow K+} g^{\prime}(x)=\lim_{x \rightarrow
K+} V^{\prime}(x),
\end{eqnarray*}
and there is no smooth fit.

\subsection{Geometric Stochastic Dynamics}
Consider now the geometric Lévy process $Y=\{Y_{t}\}$ with a finite Lévy measure
$\nu=\lambda m$ characterized by the dynamics
\begin{eqnarray} \label{glpsde}
d Y_{t}=Y_{t-}\left\{\alpha dt+\sigma d W_{t}+\lambda
\int_{\mathcal{S}}\gamma(z) \tilde{N}(dt,dz) \right\},
\end{eqnarray}
where both the drift $\alpha$ and the diffusion coefficient $\sigma$
are assumed to be positive. Note that in this case
$\mathcal{I}=\mathbb{R}_{+}$ and the explicit solution $Y_{t}$ equals
\begin{eqnarray} \label{specificglp}
y_{0}\exp\Big\{\tilde{\alpha}t+\sigma
W_{t}+\int_{0}^{t}\int_{\mathcal{S}}\ln(1+\gamma(z))\tilde{N}(ds,dz)\Big\}.
\end{eqnarray}
where $\tilde{\alpha}=\alpha-\frac{1}{2}\sigma^2$. For simplicity
of exposition, we take $\gamma(z)=-z$ and assume that $\mathcal{S}= (0,1)$.

The
integro-differential equation $(\mathcal{G}_{r}\psi)(x)=0$ takes now the form
\begin{eqnarray} \label{glpintegrode}
&&\frac{1}{2} \sigma^{2} x^2 \psi''(x)+\hat{\alpha} x
\psi'(x)-(r+\lambda)\psi(x)+\lambda
\int_{0}^{1}\psi(x-xz)\mathfrak{m}(dz)=0,\quad
\end{eqnarray}
where $\hat{\alpha}=\alpha+\lambda \bar{m}$. By guessing now the
solution to be of form $x^k$, we obtain the characteristic equation
for $k$:
\begin{eqnarray} \label{eqfork}
&&\frac{1}{2} \sigma^{2} k(k-1)+(\alpha+\lambda
\bar{m})k-(r+\lambda)+\lambda \int_{0}^{1}(1-z)^k
\mathfrak{m}(dz)=0.
\end{eqnarray}
It is straightforward to show that if $r>0$ then
\eqref{eqfork} has a positive solution $k_{1}>0$. In that case $\psi(x)=x^{k_{1}}$
is an increasing smooth solution of \eqref{glpintegrode} which
vanishes at $x=0$ and, therefore, satisfies ({\bf A1}). Moreover, if inequality $\alpha < r$ is satisfied,
then $k_1>1$. It is also at this point worth emphasizing that if $\alpha +\lambda \bar{m} > \sigma^2/2$ then \eqref{eqfork} implies that
$\lim_{r\downarrow 0}k_1 = 0$. Consequently, in that case we observe that for all $x\leq y$ we have
$$
\mathbb{P}_x[\tau_y < \infty] = \lim_{r\downarrow 0}\mathbb{E}_x\left[e^{-r\tau_y};\tau_y < \infty\right] = \lim_{r\downarrow 0}\left(\frac{x}{y}\right)^{k_1} = 1.
$$

In light
of our representation of the value of the optimal policy in terms of
an associated nonlinear programming problem, we find that for any
reward function $g$ satisfying the conditions of Theorem \ref{mainthm}, the value of
the optimal stopping policy can be represented as
\begin{eqnarray} \label{gbm1rep}
V_\lambda(x)=x^{k_{1}} \sup_{y \geq x}\left\{y^{-k_{1}}g(y)\right\}=
\begin{cases}
g(x), & x \geq x^{\ast}\\
g(x^{\ast})(x/x^\ast)^{k_1}, & x < x^{\ast},
 \end{cases}
\end{eqnarray}
where $x^{\ast}$ is the unique maximizer of $g/\psi$, i.e. for a
differentiable $g$ the solution of $g'(x^\ast)x^\ast/g(x^\ast) =
k_{1}$.

As in the arithmetic case, we observe that our Theorem
\ref{domination} implies that in the present case the root of the
equation \eqref{eqfork} $k_1$ satisfies the condition $k_1\in
(\hat{k}_{r}, \hat{k}_{r+\lambda})$, where
$$
\hat{k}_{\theta} = \frac{1}{2}-\frac{\alpha + \lambda
\bar{m}}{\sigma^2}+\sqrt{\left(\frac{1}{2}-\frac{\alpha + \lambda
\bar{m}}{\sigma^2}\right)^2+\frac{2\theta}{\sigma^2}}
$$
denotes the positive root of the characteristic equation $\sigma^{2}
k(k-1)+2(\alpha+\lambda \bar{m})k-2\theta=0$. Therefore, we observe that
$$
x^{\hat{k}_{r+\lambda}}\sup_{y\geq x}\left[g(y)y^{-\hat{k}_{r+\lambda}}\right] \leq
x^{k_1}\sup_{y\geq x}\left[g(y)y^{-k_1}\right] \leq
x^{\hat{k}_{r}}\sup_{y\geq x}\left[g(y)y^{-\hat{k}_{r}}\right]
$$
provided that the maximum exists. The jump risk adjusted discount rate $\theta^\ast$ defined in Corollary \ref{discountcor} for which $x^\ast_{\theta^\ast}=x^\ast$ reads in the present geometric setting as
$$
\theta^\ast=r+\lambda\int_{0}^{1}(1-(1-z)^{k_1})
\mathfrak{m}(dz).
$$
Alternatively, letting the drift coefficient of the associated continuous diffusion $\tilde{X}_t$
to be
$$
\tilde{\mu}(x)=(\alpha+\lambda
\bar{m}) x + \frac{\lambda x}{k_1}\int_{0}^{1}(1-z)^{k_1}
\mathfrak{m}(dz)
$$
results into the equality $x_{r+\lambda}^\ast=x^\ast$.
Finally, as demonstrated in Theorem \ref{certaintyequiv}
choosing  $\hat{\mu}_\lambda(x)=rx/k_1$ as the risk-adjusted growth rate implies that
$\hat{V}(x)=V_\lambda(x)$.

It is also clear from our analysis that the increasing fundamental
solution is strictly convex (concave) in this case as well provided that condition $r>\alpha$ ($r<\alpha$) is satisfied. Thus, as our
results in Theorem \ref{incvol} and in Theorem \ref{incintensity}
indicated, increased volatility and higher jump-intensity
increase the value and decelerate exercise timing by increasing the optimal
stopping boundary whenever $r>\alpha$. The opposite comparative static properties are satisfied when $r<\alpha$. As predicted by Theorem \ref{certaintyequiv}, the same conclusions are valid for the risk adjusted
growth rate $\hat{\mu}_\lambda(x)=rx/k_1$ as well. Moreover, since
$$
\lim_{\lambda\downarrow 0}k_1 = \frac{1}{2}-\frac{\alpha}{\sigma^2}+\sqrt{\left(\frac{1}{2}-\frac{\alpha}{\sigma^2}\right)^2+\frac{2r}{\sigma^2}}
$$
we notice that
$$
\hat{\mu}_\lambda(x)-\hat{\mu}_0(x)=rx
\left[\frac{1}{k_1}-\frac{1}{\frac{1}{2}-\frac{\alpha}{\sigma^2}+\sqrt{\left(\frac{1}{2}-\frac{\alpha                                                   }{\sigma^2}\right)^2+\frac{2r}{\sigma^2}}}\right].
$$
We illustrate the risk adjusted growth rate for $\textrm{Beta}(c,d)$-distributed jumps in Table \ref{table2} under the assumptions that $\alpha=0.03,r=0.05,c = 1.25, $ and $d = 5$.
As these numerical values indicate, jump risk has still a nontrivial effect on the required risk adjustment.
\begin{table}[!ht]
  \centering
  \begin{tabular}{|c|c|c|c|c|c|}
  \hline
  $\sigma$ & 0.05&0.1&0.15&0.2&0.25 \\
  \hline
  $\frac{\hat{\mu}_0(x)}{x}-\alpha$ & 0.08&0.27&0.49&0.7&0.89 \\
  $\frac{\hat{\mu}_{0.1}(x)}{x}-\alpha$ & 0.25&0.4&0.58&0.77&0.94 \\
  $\frac{\hat{\mu}_{0.2}(x)}{x}-\alpha$ &0.38&0.5&0.66&0.83&0.98\\
 \hline
\end{tabular}
  \caption{Risk Adjusted Growth Rates in Percentage Terms}\label{table2}
\end{table}
A case where $\psi_\lambda(x)$ is concave and the comparative statics are reversed is illustrated in Table \ref{table3} under the parameter specifications $\alpha=0.05, r=0.03, c = 1.25, $ and $d = 5$.
\begin{table}[!ht]
  \centering
  \begin{tabular}{|c|c|c|c|c|c|}
  \hline
  $\sigma$ & 0.05&0.1&0.15&0.2&0.25 \\
  \hline
  $\frac{\hat{\mu}_0(x)}{x}-\alpha$ & -0.05&-0.19&-0.39&-0.63&-0.86 \\
  $\frac{\hat{\mu}_{0.1}(x)}{x}-\alpha$ & -0.19&-0.32&-0.5&-0.72&-0.93 \\
  $\frac{\hat{\mu}_{0.2}(x)}{x}-\alpha$ &-0.32&-0.44&-0.6&-0.79&-0.98\\
 \hline
\end{tabular}
  \caption{Risk Adjusted Growth Rates in Percentage Terms}\label{table3}
\end{table}

In order to present an explicit illustration, let
$g(x)=\max(a x^b-K,0)$ with $a,K>0$ and $b\in(0,1]$. This case contains the standard
American call option (take $a=b=1$) as well as the rewards of many
optimal stopping problems associated with irreversible investment
decisions (see \citet{boyarchenko} for a very readable account on
the relationship between perpetual American options and irreversible
investment decisions).  If $k_1 > b$, then the function $g/\psi$ attains a unique maximum
at
$$
x^\ast = \left(\frac{k_1K}{(k_1-b)a}\right)^{1/b}.
$$
By Theorem \ref{mainthm}, the value of the
optimal stopping problem can now be represented as
\begin{eqnarray*}
&&V_\lambda(x)=x^{k_{1}} \sup_{y \geq x}\left\{y^{-k_{1}}(a y^b-K)\right\} =
\begin{cases}
a x^b-K, & x \geq x^{\ast}\\
(a {x^\ast}^b-K)(x/x^{\ast})^{k_{1}}, & x <
            x^{\ast}
\end{cases}
\end{eqnarray*}
provided that condition $k_1>b$ is satisfied. As usually in the real options
literature on irreversible investment, we notice that the option
multiplier $P=k_1/(k_1-b)$ determines the sensitivity of the optimal exercise threshold $x^\ast$ with respect to changes in volatility. This multiplier reads
as $\hat{P}_{\theta}=\hat{k}_\theta/(\hat{k}_\theta-b)$,$\theta=r, r+\lambda$,  for the stopping problems of the
associated continuous diffusion. We illustrate these option
multipliers in the convex setting where $r>\alpha$ in Figure \ref{kuva1} for $\textrm{Beta}(c,d)$-distributed jumps
under the assumption that $\alpha = 0.025, r = 0.05, \lambda = 0.02,
a = b = K=1, c = 1.25, d = 5$.
\begin{figure}[!ht]
\begin{center}
\includegraphics{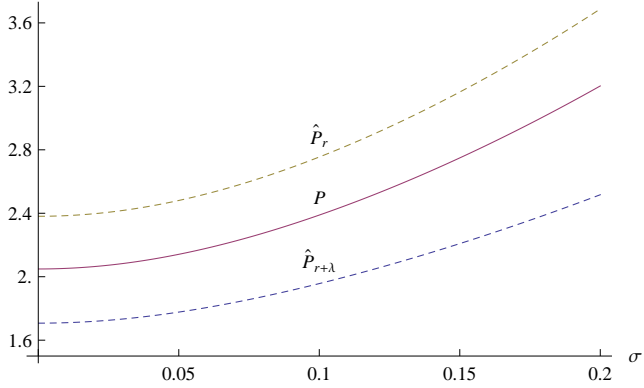}
\end{center}
\caption{The impact of volatility on the option multipliers $P$,
$\hat{P}_{r+\lambda}$, and $\hat{P}_r$}\label{kuva1}
\end{figure}
As Figure \ref{kuva1} indicates, the option multipliers are increasing as
functions of the underlying volatility coefficient. Moreover, the
option multipliers satisfies the condition $P\in
(\hat{P}_{r+\lambda},\hat{P}_{r})$ as was established in our Theorem
\ref{domination}. The values of the optimal stopping problems are
graphically illustrated for $\textrm{Beta}(c,d)$-distributed jumps
in Figure \ref{kuva2} under the assumption that $\alpha = 0.025, r = 0.05, \lambda = 0.02,
a = b = K=1, c = 1.25, d = 5$, and  $\sigma=0.1$ (which
implies that $x^\ast = P =2.39, x^\ast_{0.07} = \hat{P}_{0.07} = 1.96$, and
$x^\ast_{0.05} = \hat{P}_{0.05} = 2.75$)
\begin{figure}[!ht]
\begin{center}
\includegraphics{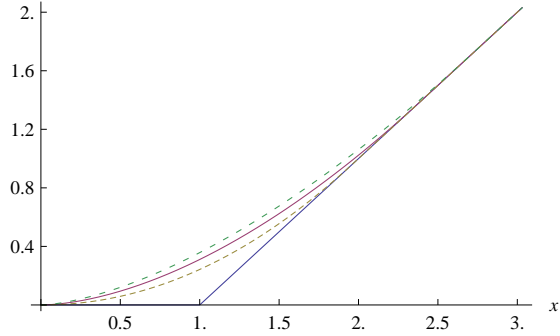}
\end{center}
\caption{The exercise payoff $(x-1)^+$ and the values $V_\lambda(x),
\tilde{V}_r(x),$ and $\tilde{V}_{r+\lambda}(x)$}\label{kuva2}
\end{figure}
Figure \ref{kuva2} illustrates explicitly the results of our Theorem
\ref{domination} for the values of the stopping problems. It is of
interest to notice that as was predicted by Theorem
\ref{domination}, the value $V_\lambda(x)$ of the considered stopping
problem is sandwiched between the two values
$\tilde{V}_{r+\lambda}(x)$ and $\tilde{V}_{r}(x)$.

It is worth noticing that in the present
example the maximizing threshold $x^\ast$ exists even in cases where $k_1<1$,
that is, even when the fundamental solution is not convex as a
function of the state. Hence, for the considered exercise payoff the condition  $r>\alpha$ guaranteeing the convexity of $x^{k_1}$ can be relaxed. If $0 < r\leq
\alpha$ then $k_1\in (0,1]$. Under those circumstances the sign of the relationship
between increased volatility and the optimal exercise strategy is
reversed as the root $k_1$ becomes an increasing function of
volatility. More precisely, if $r\leq\alpha$ then
 $\partial x^\ast/\partial\sigma =
-(x^\ast/k_1^2) \partial k_1/\partial\sigma < 0$ and $\partial
x^\ast/\partial\lambda = -(x^\ast/k_1^2) \partial
k_1/\partial\lambda < 0$. We illustrate this observation graphically
for $\textrm{Beta}(c,d)$-distributed jumps in Figure \ref{kuva3} under the
assumption that $\alpha = 0.04, r = 0.02, \lambda = 0.01, a = K = 1, b=0.2,
c = 1.25$, and $d = 2$.
\begin{figure}[!ht]
\begin{center}
\includegraphics{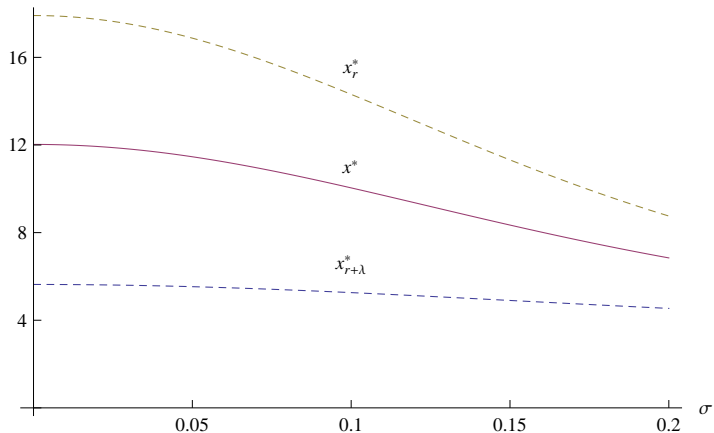}
\end{center}
\caption{The impact of volatility on the exercise thresholds}\label{kuva3}
\end{figure}
Figure \ref{kuva3} illustrates how the sign of the relationship between
increased volatility and the optimal exercise threshold is reversed
as the increasing fundamental solution becomes concave. It is worth
noticing that even in this case the order of the exercise thresholds
remain naturally unchanged since the ordering of the values $V_\lambda(x)$,
$\tilde{V}_{r+\lambda}(x)$, and $\tilde{V}_{r}(x)$ is based only on
nonnegativity and monotonicity.

\section{Conclusions}

In this study we generalized a representation result known to hold
for continuous linear diffusions to include a class of spectrally
one-sided L\'{e}vy diffusions: given some conditions, the optimal
stopping problem for a one-dimensional spectrally negative L\'{e}vy
diffusion can be reduced to an ordinary nonlinear programming
problem.

Considering the fact that optimal stopping problems feature
prominently in pricing of American options and in real options
theory, reducing the stopping problem of a L\'{e}vy diffusion into a
standard programming problem can significantly facilitate the
ongoing research on these areas of mathematical finance. We
demonstrated this by deriving several interesting comparative static
properties of spectrally negative L\'{e}vy diffusions using our
representation, and found out that a useful tool in obtaining bounds
for the value of the optimal stopping of a L\'{e}vy diffusion is the
corresponding stopping problem for an associated continuous
diffusion. By choosing the discount rates appropriately, we were able
to sandwich the value of the considered optimal stopping
problem between the known values of two stopping
problems of the associated continuous diffusion.

Our study indicates that
the impact of volatility on the optimal policy and its value in our
setting is similar to the continuous case: for values convex
(concave) below the optimal threshold, increased risk decelerates
(accelerates) rational investment by expanding or leaving unchanged
(shrinking or leaving unchanged) the continuation region and
increasing or leaving unchanged (decreasing or leaving unchanged)
the optimal threshold and the value of waiting. The impact of
downside risk as measured by the intensity of the compound Poisson
jump process on the optimal value was found out to be similar to the
impact of the diffusion risk (as measured by the volatility). We
also established that the key factor determining the relevant
convexity/concavity properties of the value is (provided that it
exists) the increasing fundamental solution of the associated
integro-differential equation, which is process-specific. Thus we
saw that the impact of volatility or downside risk is not dependent
on the precise form of the exercise payoff, as long as the
conditions for the optimality of the stopping rule characterized by
a single threshold are met.

In addition to their usefulness in obtaining information about the
comparative static properties of L\'{e}vy diffusions and their
relations (similarities and differences) to the continuous diffusion
case, our results raise an interesting question on the scope of
applicability of our representation. This boils largely down to the
question: when is the assumption on the existence of an increasing
smooth solution to the characteristic integro-differential equation
true, and can conveniently verifiable sufficient conditions for this
be found? The technique developed in \citet{Rakkolainen2008} based
on Frobenius series solutions appears to be a promising approach which may provide
further insights into the considered class of problems.\\

\noindent{\bf Acknowledgments:}
The authors wish to thank \emph{Erik Baurdoux} and
\emph{Olli Wallin} for their insightful comments. The financial
support to Luis H. R. Alvarez E. from the {\it OP Bank Research Foundation} is gratefully
acknowledged.

\end{document}